\documentclass[aps,amsmath,amssymb,twocolumn,nofootinbib,pra,superscriptaddress]{revtex4}

\usepackage{graphicx}
\usepackage{dcolumn}
\usepackage{bm}
\usepackage{epsfig,hyperref,amsthm}
\usepackage{mathtools}
\usepackage{color,soul}
\usepackage{colordvi}
\usepackage[dvipsnames]{xcolor}
\usepackage{relsize}
\usepackage{accents}
\usepackage{wasysym}  
\usepackage{xypic}
\usepackage{enumitem}

\hypersetup{
	colorlinks=true,
	linkcolor=blue,
	citecolor=blue,
	urlcolor=blue
}

\newcommand{\ket}[1]{\ensuremath{|#1\rangle}}
\newcommand{\bra}[1]{\ensuremath{\langle #1|}}

\newcommand{\proj}[1]{\ket{#1}\bra{#1}}
\newcommand{\be}{\begin{equation}}
\newcommand{\ee}{\end{equation}}
\newcommand{\ba}{\begin{eqnarray}}
\newcommand{\ea}{\end{eqnarray}}

\newcommand{\norm}[1]{\left\|#1\right\|}

\newcommand{\id}{\mathbb{I}}

\newcommand{\F}[1]{\mathcal{F}_{\rm max}(#1)}

\newtheorem{theorem}{Theorem}
\newtheorem{corollary}[theorem]{Corollary}

\newtheorem{alemma}{Lemma}[section]
\newtheorem{aproposition}[alemma]{Proposition}
\newtheorem{afact}[alemma]{Fact}

\newtheorem{atheorem}[alemma]{Theorem}
\newtheorem{adefinition}[alemma]{Definition}

\newtheorem{question}{Question}

\newcommand{\neweta}{\wt{\gamma}}

\newcommand{\wt}[1]{\widetilde{#1}}
\newcommand{\mE}{\mathcal{E}}

\definecolor{nred}{rgb}{0.9,0.1,0.1}
\definecolor{nblack}{rgb}{0,0,0}
\definecolor{nblue}{rgb}{0.2,0.2,0.8}
\definecolor{ngreen}{rgb}{0.2,0.6,0.2}
\definecolor{ublue}{rgb}{0,0,0.5}
\definecolor{OliveGreen}{cmyk}{0.64,0,0.95,0.40}
\definecolor{pur}{rgb}{0.75,0,0.75}
\definecolor{nngrn}{rgb}{0,0.5,0.5}

\newcommand{\blu}{\color{nblue}}

\begin{document}
\title{Communication, Dynamical Resource Theory, and Thermodynamics}

\author{Chung-Yun Hsieh}
\email{chung-yun.hsieh@icfo.eu}
\affiliation{ICFO - Institut de Ci\`encies Fot\`oniques, The Barcelona Institute of Science and Technology, 08860 Castelldefels, Spain}

\date{\today}

\begin{abstract} 
Recently, new insights have been obtained by jointly studying communication and resource theory.
This interplay consequently serves as a potential platform for interdisciplinary studies.
To continue this line, we analyze the role of dynamical resources in a communication setup, and further apply our analysis to thermodynamics.
To start with, we study classical communication scenarios constrained by a given resource, in the sense that the information processing channel is unable to supply additional amounts of the resource.
We show that the one-shot classical capacity is upper bounded by resource preservability, which is a measure of the ability to preserve the resource. 
A lower bound can be further obtained when the resource is asymmetry.
As an application, unexpectedly, under a recently-studied thermalization model, we found that the smallest bath size needed to thermalize all outputs of a Gibbs-preserving coherence-annihilating channel upper bounds its one-shot classical capacity.
When the channel is coherence non-generating, the upper bound is given by a sum of coherence preservability and the bath size of the channel's incoherent version.
In this sense, bath sizes can be interpreted as the thermodynamic cost of transmitting classical information.
This finding provides a dynamical analogue of Landauer's principle, and therefore bridges classical communication and thermodynamics.
As another implication, we show that, in bipartite settings, classically correlated local baths can admit classical communication even when both local systems are completely thermalized.
Hence, thermalizations can transmit information by accessing only classical correlation as a resource.
Our results demonstrate interdisciplinary applications enabled by dynamical resource theory.
\end{abstract}

\maketitle

\section{Introduction}\label{Sec:Intro}
{\em Resource} is a concept widely used in the study of physics: It can be an effect or a phenomenon, helping us achieve advantages that can never occur in its absence.
A quantitative understanding of different resources is thus vital for further applications.
For this reason, an approach called {\em resource theory} comes, aiming to provide a general strategy to depict different resources~\cite{RT-RMP}.

A resource theory can be interpreted as a triplet $(R, \mathcal{F}_R, \mathcal{O}_R)$, consisting of the resource itself $R$ (e.g., entanglement~\cite{Ent-RMP}), the set of quantities without the resource $\mathcal{F}_R$ (called {\em free quantities}; e.g., separable states), and the set of physical processes that will not generate the resource $\mathcal{O}_R$ (called {\em free operations} of $R$; e.g., local operation and classical communication channels~\cite{QCI-book}).
Every $\mathcal{E}\in\mathcal{O}_R$ must satisfy $\mathcal{E}(\eta)\in\mathcal{F}_R\;\forall\eta\in\mathcal{F}_R$, which is sometimes called the golden rule of resource theories~\cite{RT-RMP}.
Operations satisfying this condition are called {\em resource non-generating} for $R$, which form the largest possible set of free operations of $R$.
A resource theory allows ones to quantify the resource via a {\em resource monotone}, $Q_R$, which is a non-negative-valued function satisfying two conditions:
(i) $Q_R(q) = 0$ if $q\in\mathcal{F}_R$; and
(ii) $Q_R[\mathcal{E}(q)]\le Q_R(q)\;\forall q\;\&\;\forall\mathcal{E}\in\mathcal{O}_R$.
This is a ``ruler'' attributing numbers to different resource contents.

Adopting this general approach, one can study specific resources such as (but not limited to) entanglement~\cite{Ent-RMP,Vedral1997,Vidal2002}, coherence ~\cite{Coherence-RMP,Baumgratz2014}, nonlocality~\cite{Bell,Bell-RMP,Wolfe2019}, steering~\cite{Wiseman2007,Jones2007,steering-review,Skrzypczyk2014,Piani2015,Gallego2015,RMP-steering}, asymmetry~\cite{Gour2008,Marvian2016,Takagi2019-4}, and athermality~\cite{Brandao2013,Brandao2015,Horodecki2013,Lostaglio2018,Serafini2019,Narasimhachar2019}.  
Together with various features of general resource theories~\cite{Horodecki2013-2,Brandao2015-2,del_Rio2015,Coecke2016,Gour2017,Anshu2018,Regula2018,Bu2018,Liu2017,Lami2018,RT-RMP,Takagi2019,Takagi2019-2,Liu2019,Fang2019,Korzekwa2019,Regula2020}, one is able to concretely picture the originally vague notion of resources for {\em states} -- while the unique roles of dynamical resources have not been noticed until recently. 
Resource theories of {\em channels}~\cite{footnote0} have therefore drawn much attention lately and been studied intensively~\cite{Hsieh2017,Kuo2018,Pirandola2017,Dana2017,Bu2018,Wilde2018,Diaz2018,Zhuang2018,Gour2019-3,Theurer2019,Seddon2019,Rosset2018,LiuWinter2019,LiuYuan2019, Gour2019,Gour2019-2,Bauml2019,Wang2019,Berk2019,Takagi2019-3,Hsieh2020-1,Saxena2019,Zhang2020}. 
Unlike the state resources, which are static, channel resources are dynamical properties, thereby providing links to dynamical problems such as communication~\cite{Takagi2019-3} and resource preservation~\cite{Hsieh2020-1,Saxena2019}.

Very recently, the interplay between resource theories and classical communication has been investigated~\cite{Korzekwa2019,Takagi2019-3} (see also Ref.~\cite{Kristjansson2020}), successfully providing new insights and widening our understanding.
For instance, a neat proof of the strong converse property of non-signaling assisted classical capacity has been established~\cite{Takagi2019-3}.
Also, amounts of classical messages encodable into the resource content of states has been estimated, and different physical meanings can be concluded by considering specific resources~\cite{Korzekwa2019}.
Hence, the interplay between resource theory and classical communication is a potential platform for interdisciplinary studies.
To continue this research line, it is thus necessary to understand communication setups constrained by different static resources.
A general treatment on this can clarify the role of static resources in communication and provide potential applications in different physical settings.
This motivates us to ask:
\begin{center}
{\em
How do resource constraints affect classical communications?
}
\end{center}
In this work, we consider classical communication scenarios where the information processing channel is forbidden to supply additional resources, thereby being a free operation (there are some subtleties about this setting, and we refer the reader to Appendix~\ref{App:NotionCostless} for a detailed discussion).
The basic setup will be given in Sec.~\ref{Sec:Formulation}.
In Sec.~\ref{Sec:Capacity}, we show that the corresponding one-shot classical capacity is upper bounded by the ability to preserve the resource, which is called resource preservability \cite{Hsieh2020-1}, plus a resourceless contribution term.
Furthermore, when the underlying resource is asymmetry, a lower bound can be obtained.
As an application, we use our approach to bridge classical communication and thermodynamics in Sec.~\ref{Sec:Thermo}:
Under the thermalization model introduced in Ref.~\cite{Sparaciari2019}, the one-shot classical capacity of a Gibbs-preserving coherence non-generating channel is upper bounded by its coherence preservability plus the smallest bath size needed to thermalize all outputs of its incoherent version~\cite{Sparaciari2019,Hsieh2020-1}.
A direct physical message is that, under this setting, transmitting classical information through a coherence-annihilating channel necessarily needs a large enough bath size as the thermodynamic cost. 
This provides a dynamical analogue of the famous Landauer's principle~\cite{Landauer1961} and 
illustrates how dynamical resource theory can connect seemingly different fields, inspiring us to further investigate the interplay of classical communication and thermalization.
To this end, in Sec.~\ref{Sec:MaintainingMaxEnt} we first study a tool related to the ability to simultaneously maintain orthogonality and maximal entanglement.
Using this concept, in Sec.~\ref{Sec:CCThermalization} we found that it is possible for locally performed thermalization channels to globally transmit classical information when the local baths are correlated classically through pre-shared randomness.
This result reveals the huge difference between local and global dynamics, providing the recent discovery in Ref.~\cite{Hsieh2020-2} a generalization and application to communication.
Finally, we conclude in Sec.~\ref{Sec:Conclusion}.

\section{Formulation}\label{Sec:Formulation}
To process {\em classical information} depicted by a finite sequence of integers $\{m\}_{m=0}^{M-1}$ through a quantum channel $\mathcal{N}$, one needs to encode them into a set of quantum states $\{\rho_m\}_{m=0}^{M-1}$; likewise, a decoding is needed to extract the information from outputs of $\mathcal{N}$, which can be done by a {\em positive operator-valued measurement} (POVM) $\{E_m\}_{m=0}^{M-1}$~\cite{QCI-book}.
They can be written jointly as $\Theta_M = (\{\rho\}_{m=0}^{M-1},\{E_m\}_{m=0}^{M-1})$, called an {\em $M$-code},
which depicts the transformation
$
\rho_m\mapsto {\rm tr}\left[E_m\mathcal{N}(\rho_m)\right]
$
for each $m$.
To see how faithfully one can extract the input messages $\{m\}_{m=0}^{M-1}$, the {\em one-shot classical capacity with error $\epsilon$}~\cite{Wang2012,Datta2013} of $\mathcal{N}$ can be defined as a measure:
\begin{align}\label{Eq:Capacity}
C_{\rm (1)}^\epsilon(\mathcal{N})\coloneqq\max\left\{\log_2{M}\;|\;\exists\Theta_M,\;p_s(\Theta_M,\mathcal{N})\ge1-\epsilon\right\},
\end{align}
where the {\em average success probability} is given by
\begin{align}\label{Eq:SuccessProbability}
p_s(\Theta_M,\mathcal{N})\coloneqq\frac{1}{M}\sum_{m=0}^{M-1}{\rm tr}\left[E_m\mathcal{N}(\rho_m)\right].
\end{align}

Before introducing the main results, we briefly review relevant ingredients of resource preservability~\cite{Hsieh2020-1} (or simply {\em $R$-preservability} when the state resource $R$ is given), which is a dynamical resource depicting the ability to preserve $R$.
To start with, we impose basic assumptions on a given state resource theory $(R,\mathcal{F}_R,\mathcal{O}_R)$:
\begin{enumerate}
\item\label{Postulate:Identity} Identity and partial trace are both free operations.
\item\label{Postulate:Composition} $\mathcal{O}_R$ is closed under tensor products, convex sums, and compositions ~\cite{footnote:Markovianity}.
\end{enumerate}
These assumptions are strict enough for an analytically feasible study and still general enough to be shared by many known resource theories (see Appendix~\ref{App:Basic-Assumptions} for a further discussion).
For a given state resource theory $(R,\mathcal{F}_R,\mathcal{O}_R)$, the induced $R$-preservability theory is a channel resource theory written as ($R$-preservability, $\mathcal{O}_R^N$, $\mathbb{F}_R$).
It is defined on all channels in $\mathcal{O}_R$.
In this channel resource theory, the free quantities are members of the set $\mathcal{O}_R^N$ called {\em resource annihilating channels}~\cite{Hsieh2020-1,footnote:R-DestroyingMap}:
\begin{align}
\mathcal{O}_R^N\coloneqq\{\Lambda\in\mathcal{O}_R\;|\;\Lambda(\rho)\in\mathcal{F}_R\;\forall\rho\}.
\end{align}
They are channels in $\mathcal{O}_R$ which can only output free states.
A special class of resource annihilating channels are those who cannot output any resourceful state even assisted by ancillary resource annihilating channels; specifically, no $R$-preservability can be activated~\cite{Palazuelos2012,Cavalcanti2013,Hsieh2016,Quintino2016,Hsieh2020-1,Zhang2020}. 
Such channels are called {\em absolutely resource annihilating channels}~\cite{Hsieh2020-1}, which are elements of the set
$
\wt{\mathcal{O}}_R^N\coloneqq\{\wt{\Lambda}\in\mathcal{O}_R^N\;|\;\wt{\Lambda}\otimes\Lambda\in\mathcal{O}_R^N\;\forall\Lambda\in\mathcal{O}_R^N\}.
$
Free operations of $R$-preservability, which are collectively denoted by the set $\mathbb{F}_R$, are {\em super-channels}~\cite{Chiribella2008,Chiribella2008-2} given by~\cite{Hsieh2020-1}
$
\mE\mapsto\Lambda_+\circ(\mE\otimes\wt{\Lambda})\circ\Lambda_-
$
with $\Lambda_+,\Lambda_-\in\mathcal{O}_R$ and $\wt{\Lambda}\in\wt{\mathcal{O}}_R^N$.
Finally, to quantify $R$-preservability, consider a contractive generalized distance measure $D$ on quantum states; that is, it is a real-valued function satisfying (i) $D(\rho,\sigma)\ge0$ and equality holds if and only if $\rho=\sigma$, and (ii) (data-processing inequality) $D[\mE(\rho),\mE(\sigma)]\le D(\rho,\sigma)$ for all $\rho,\sigma$ and channels $\mE$.
The following $R$-preservability monotone induced by $D$~\cite{Hsieh2020-1} will be used in this work: 
\begin{align}\label{Eq:P_D}
P_{D}(\mathcal{E})\coloneqq\inf_{\Lambda\in\mathcal{O}_R^N}D^R(\mathcal{E},\Lambda),
\end{align}
where 
$D^R(\mathcal{E}_{\rm S},\Lambda_{\rm S})\coloneqq\sup_{\wt{\Lambda}_{\rm A},\rho_{\rm SA}} D[(\mathcal{E}_{\rm S}\otimes\widetilde{\Lambda}_{\rm A})(\rho_{\rm SA}),(\Lambda_{\rm S}\otimes\widetilde{\Lambda}_{\rm A})(\rho_{\rm SA})]$
maximizes over every possible ancillary system ${\rm A}$, joint input $\rho_{\rm SA}$, and absolutely resource annihilating channel $\wt{\Lambda}_{\rm SA}$.
Geometrically, Eq.~\eqref{Eq:P_D} can be understood as a distance between $\mathcal{E}$ and the set $\mathcal{O}_R^N$ that is adjusted by absolutely resource annihilating channels.
With Assumptions~\ref{Postulate:Identity} and~\ref{Postulate:Composition}, in Appendix~\ref{App:Basic-Assumptions} we show that $P_D$ is indeed a monotone~\cite{Hsieh2020-1}, in the sense that it is a non-negative-valued function such that $P_D(\mE) = 0$ if $\mE\in\mathcal{O}_R^N$~\cite{footnote:FaithfulnessRemark}, and $P_D[\mathfrak{F}(\mE)]\le P_D(\mE)$ for every channel $\mE$ and $\mathfrak{F}\in\mathbb{F}_R$.
Note that, in this work, we only ask $R$-preservability monotones to satisfy these two conditions, which are the core features of a monotone.
Further properties, such as Eq.~(7) in Ref.~\cite{Hsieh2020-1}, will need additional assumptions, and we leave the details in Appendix~\ref{App:Basic-Assumptions}.
Finally, the distance measure that will be mainly used in this work is the {\em max-relative entropy}~\cite{Datta2009} for states defined as (and, conventionally, we adopt $\inf\emptyset = \infty$)
\begin{align}\label{Eq:Def-D_max}
D_{\rm max}(\rho\|\sigma)\coloneqq \log_2\inf\{\lambda\ge0\,|\,\rho\le\lambda\sigma\}.
\end{align}
Hence, $P_{D_{\rm max}}(\mathcal{N})$ is the minimal amount of noise one needs to add to turn $\mathcal{N}$ into resource annihilating [see also Eq.~\eqref{Eq:AlternativeFormDmaxR}].
For a given error $\kappa>0$, we define $\mathcal{O}_R^N(\kappa;\mathcal{N})\coloneqq\{\Lambda\in\mathcal{O}_R^N\,|\,|D_{\rm max}^R(\mathcal{N}\|\Lambda) - P_{D_{\rm max}}(\mathcal{N})|\le\kappa\}$.
It contains resource annihilating channels achievable by adding the smallest amount of noise to $\mathcal{N}$, up to the error $\kappa$.
We call them {\em resourceless versions} of $\mathcal{N}$ up to the error $\kappa$, and use the notation $\Lambda^\mathcal{N}\in\mathcal{O}_R^N(\kappa;\mathcal{N})$ to emphasize their dependence on $\mathcal{N}$.

\section{Bounds On Classical Capacity}\label{Sec:Capacity}
To introduce the first result, define
\begin{align}\label{Eq:GammaQuantity}
\Gamma_\kappa(\mathcal{N})\coloneqq\log_2\inf_{\Lambda^\mathcal{N}\in\mathcal{O}_R^N(\kappa;\mathcal{N})}\sup_{\Theta_M}\sum_{m=0}^{M-1}{\rm tr}\left[E_m\Lambda^\mathcal{N}(\rho_m)\right],
\end{align}
where $\sup_{\Theta_M}$ maximizes over every $M\in\mathbb{N}$ and $M$-code $\Theta_M$.
$2^{\Gamma_\kappa(\mathcal{N})}$ tells us the highest number of discriminable states through every resourceless version of $\mathcal{N}$, up to the error $\kappa$.
We also consider $P_{D}^\delta(\mathcal{E})\coloneqq\inf_{\norm{\mE - \mE'}_\diamond\le2\delta}P_{D}(\mathcal{E}')$ and $\Gamma_\kappa^{\delta}(\mathcal{E})\coloneqq\sup_{\norm{\mE - \mE'}_\diamond\le2\delta}\Gamma_\kappa^{\delta}(\mathcal{E}')$, which smooth the optimizations over all channels $\mE'$ closed to $\mE$.
Also, $\norm{\mE}_\diamond\coloneqq\sup_{{\rm A},\rho_{\rm SA}}\norm{(\mE\otimes\mathcal{I}_{\rm A})(\rho_{\rm SA})}_1$ is the {\em diamond norm}.
Combining Refs.~\cite{Takagi2019-3,Hsieh2020-1}, in Appendix~\ref{App:Proof-Result:Upper-Bound} we prove the following upper bound:
\begin{theorem}\label{Result:Upper-Bound}
Given $\epsilon,\delta\ge0\;\&\;0<\kappa<1$ satisfying $\epsilon+\delta<1$. 
Then for every $\mathcal{N}\in\mathcal{O}_R$ we have
\begin{align}\label{Eq:Result:CCUpper-Bound}
C_{\rm (1)}^\epsilon(\mathcal{N})\le P_{D_{\rm max}}^\delta(\mathcal{N})  + {\Gamma_\kappa^{\delta}(\mathcal{N})} + \log_2\frac{2^\kappa}{1-\epsilon-\delta}.
\end{align}
\end{theorem}
This upper bound Eq.~\eqref{Eq:Result:CCUpper-Bound} can be interpreted as follows: It is the highest amount of carriable classical information by every resourceless version of $\mathcal{N}$, i.e., ${\Gamma_\kappa^{\delta}}(\mathcal{N})$, plus the contribution from the ability of $\mathcal{N}$ to preserve $R$, i.e., $P_{D_{\rm max}}^\delta(\mathcal{N})$.
This also suggests that $C_{\rm (1)}^\epsilon(\mathcal{N}) - {\Gamma_\kappa^{\delta}(\mathcal{N})}$ characterizes the resource advantage, since it estimates the amount of transmissible classical information via the ability of $\mathcal{N}$ to preserve $R$.
To see the tightness of Eq.~\eqref{Eq:Result:CCUpper-Bound} (up to error terms containing $\epsilon,\delta$), the inequality is saturated by every state preparation channel outputting a fixed free state $(\cdot)\mapsto\eta$ with $\eta\in\mathcal{F}_R$.
There also exist examples attending the inequality with non-zero classical capacity.
For instance, in a $d$-dimensional system, when $R$ is coherence and $\mathcal{N}$ is the dephasing channel, i.e., $(\cdot)\mapsto\sum_{i=1}^d\proj{i}\cdot\proj{i}$, we have both side as $\log_2d$.

As the last remark on Eq.~\eqref{Eq:Result:CCUpper-Bound}, when the optimal amount of classical information can be encoded into free states, the optimal capacity should intuitively be achievable by channels with zero $R$-preservability, and any general result should respect this fact.
Hence, resource preservability monotone cannot upper bound classical capacity solely, and Theorem~\ref{Result:Upper-Bound} is consistent with this expectation due to the term $\Gamma_\kappa^\delta(\mathcal{N})$. 
However, such resourceless advantages no longer exist when more constraints are made for specific purposes (e.g., Sec.~\ref{Sec:MaintainingMaxEnt}).

Theorem~\ref{Result:Upper-Bound} can link different physical properties to classical communication, which is illustrated by the following result.
Let $P_{D|R}$ denote Eq.~\eqref{Eq:P_D} with the state resource $R$ and
write $R=\gamma$ when the state resource is the athermality induced by the thermal state $\gamma$ (we postpone its formal definition to Sec.~\ref{Sec:Thermo}).
Then in Appendix~\ref{App:Proof-Result:Upper-Bound} we show that:
\begin{corollary}\label{Coro:CCUpper-Bound}
Given $0\le\epsilon<1\;\&\;0<\kappa<1$ and a full-rank thermal state $\gamma$.
For $\mathcal{N}\in\mathcal{O}_R$ that is also Gibbs-preserving and every $\Lambda^\mathcal{N}\in\mathcal{O}_R^N(\kappa;\mathcal{N})$, we have
\begin{align}
C_{\rm (1)}^\epsilon(\mathcal{N})\le P_{D_{\rm max}|R}(\mathcal{N})  + P_{D_{\rm max}|\gamma}\left(\Lambda^\mathcal{N}\right) + \log_2\frac{2^\kappa}{1-\epsilon}.
\end{align}
\end{corollary}
Corollary~\ref{Coro:CCUpper-Bound} implies that the ability of a Gibbs-preserving $\mathcal{N}\in\mathcal{O}_R$ to transmit classical information is limited by its ability to preserve $R$, plus the ability of its resourceless version (to $R$) to preserve athermality.
This observation will allow us to prove one of the main results of this work in Sec.~\ref{Sec:Thermo}.

\subsection{Asymmetry and Lower Bounds}
When the underlying state resource is the asymmetry of a given unitary group $G$, a lower bound on the classical capacity can also be obtained.
Formally, asymmetry of a given group $G$, or simply {\em $G$-asymmetry}, has free states as those invariant under group actions; that is, $\rho = U\rho U^\dagger$ for all $U\in G$.
One option of free operations, which is adopted here, is {\em $G$-covariant channels}, which are channels commuting with unitaries in $G$: $U\mE(\cdot)U^\dagger = \mE[U(\cdot)U^\dagger]$ for all $U\in G$ (see, e.g., Ref.~\cite{MarvianThesis,Marvian2013,QRF-RMP,Takagi2019-4}). 

To introduce the result, we need to use the {\em information spectrum relative entropy}~\cite{Hayashi2003,Tomamichel2013} (see also Ref.~\cite{Korzekwa2019}) given by
$
D_s^\delta(\rho\|\sigma)\coloneqq \sup\{\omega\,|\,{\rm tr}\left(\rho\Pi_{\rho\le2^\omega\sigma}\right)\le\delta\},
$
where $\Pi_{\rho\le2^\omega\sigma}$ is the projection onto the union of eigenspaces of $2^\omega\sigma - \rho$ with non-negative eigenvalues~\cite{Korzekwa2019}.
Despite its name, the information spectrum relative entropy is not a proper contractive distance measure, since it will not satisfy data-processing inequality and can output negative values~\cite{Leditzky-PhD}.
However, it allows us to obtain a lower bound on $C_{\rm (1)}$. 
In Appendix~\ref{App:Proof-Result:Asymmetry-Lower-Bounds}, we apply results in Ref.~\cite{Korzekwa2019} and show the following bound (now $\mathcal{O}_R^N$ denotes the set of $G$-covariant channels that cannot preserve any $G$-asymmetry):
\begin{theorem}\label{Result:Asymmetry-Lower-Bounds}
Given $R=$ $G$-asymmetry, then for every $G$-covariant channel $\mathcal{N}$ and $0\le\delta<\epsilon<1$, we have
\begin{align}
\max\left\{0;\frac{1}{\ln{2}}\wt{P}_{D_s^{\epsilon - \delta}}(\mathcal{N}) + \log_2\delta -1\right\} \le C_{\rm  (1)}^\epsilon(\mathcal{N}),
\end{align}
where $\wt{P}_{D_s^{\epsilon - \delta}}(\mathcal{N})\coloneqq\inf_{\Lambda\in\mathcal{O}_R^N}\sup_{\rho}D_s^{\epsilon - \delta}[\mathcal{N}(\rho)\,\|\,\Lambda(\rho)]$.
\end{theorem}
This provides an $R$-preservability-like lower bound on the one-shot classical capacity for $G$-covariant channels, which also shows a witness of resourceful advantages, i.e., $C_{\rm (1)}^\epsilon(\mathcal{N}) - {\Gamma_\kappa^{\delta}(\mathcal{N})}$.
Using $(U\otimes U^*)$-asymmetry as an example (see also Appendix~\ref{App:ExIsotropicStateAsymm} for more details), the advantage from asymmetry can be witnessed when $\wt{P}_{D_s^{\epsilon - \delta}}(\mathcal{N})>2\ln{2}+\ln\frac{d^2}{d^2 - 1} - \ln\delta$, which is approximately $\wt{P}_{D_s^{\epsilon - \delta}}(\mathcal{N})>2\ln2 - \ln\delta$ when $d\gg1$.

\section{Application: Classical Communications And Thermodynamics}\label{Sec:Thermo}
It is worth mentioning that our result bridges two seemingly different physical concepts: Classical capacity~\cite{Takagi2019-3} and heat bath size needed for thermalization~\cite{Sparaciari2019,Hsieh2020-1}.
To introduce the result, we give a quick review of the resource theory of athermality and related ingredients for thermalization bath sizes~\cite{Sparaciari2019}.
{\em Athermality} is the status out of thermal equilibrium.
With a fixed system dimension $d$, the unique free state is the thermal equilibrium state, or the {\em thermal state}. 
With a given system Hamiltonian $H_{\rm S}$ and temperature $T$, the thermal state is uniquely given by
$
\gamma = \frac{e^{-\beta H_{\rm S}}}{{\rm tr}(e^{-\beta H_{\rm S}})},
$
where $\beta = \frac{1}{k_BT}$ is the inverse temperature and $k_B$ is the Boltzmann constant.
For multiple systems with tensor product, all free states in this resource theory are $\gamma^{\otimes k}$ for some positive integer $k$ (i.e., all allowed dimensions are $d^k$ with some $k$).
In this work, we adopt {\em Gibbs-preserving channels} as the free operations. 
They are channels $\mathcal{E}$ keeping thermal states invariant: $\mathcal{E}(\gamma^{\otimes k}) = \gamma^{\otimes l}$, where $d^k$ and $d^l$ are the input and output dimensions, respectively.
Physically, these are dynamics that will not drive thermal equilibrium away from equilibrium.

To formally study thermalization, we follow Ref.~\cite{Sparaciari2019} and define a channel (jointly acting on system ${\rm S}$ plus bath ${\rm B}$) $\mathcal{E}_{\rm SB}:{\rm SB}\to{\rm SB}$ to {\em $\epsilon$-thermalize} a system state $\rho_{\rm S}$ if
\begin{align}\label{Eq:Thermalize-Def}
\norm{\mathcal{E}_{\rm SB}\left(\rho_{\rm S}\otimes\gamma^{\otimes (n-1)}\right) - \gamma^{\otimes n}}_1\le\epsilon.
\end{align}
That is, $\mathcal{E}_{\rm SB}$ needs to globally thermalize ${\rm SB}$, where the thermal state $\gamma$ is determined by the Hamiltonian and the temperature of ${\rm S}$, and the initial state of ${\rm B}$ is the $n-1$ copies of $\gamma$.
To depict such thermalization processes dynamically, we consider the collision model introduced in Ref.~\cite{Sparaciari2019}.
To avoid complexity, we refer the reader to Appendix~\ref{App:ThermalizationModel} for a brief introduction of this model, and here we let $\mathcal{C}_n$ be the set of all channels on ${\rm SB}$ that can be realized by this model.
Then the quantity
$
n^\epsilon_\gamma(\rho_{\rm S})\coloneqq\inf\{n\in\mathbb{N}\,|\,\exists\,\mathcal{E}_{\rm SB}\in\mathcal{C}_n\;{\rm s.t.\;Eq.~\eqref{Eq:Thermalize-Def}\;holds}\}
$
can be understood as the smallest bath size needed to $\epsilon$-thermalize $\rho_{\rm S}$ under this model~\cite{Sparaciari2019}.
This concept can be generalized to any channel $\mathcal{N}$ by defining~\cite{Hsieh2020-1}
\begin{align}\label{Eq:BathSizeChannel}
\mathcal{B}^\epsilon_\gamma(\mathcal{N})\coloneqq\sup_\rho n^\epsilon_\gamma[\mathcal{N}(\rho)] - 1,
\end{align}
which maximizes over all the smallest bath sizes among all outputs of $\mathcal{N}$.
This is thus the smallest bath size needed to $\epsilon$-thermalize all outputs of $\mathcal{N}$ under the given collision model.

Now we mention a core assumption made in Ref.~\cite{Sparaciari2019} used to regularize the analysis.
A given Hamiltonian $H$ with energy levels $\{E_i\}_{i=1}^d$ is said to satisfy the {\em energy subspace condition} if for every positive integer $M$ and every pair of different vectors $\{{\bf m}\neq{\bf m}'\}\subset\left(\mathbb{N}\cup\{0\}\right)^d$ satisfying $\sum_{i=1}^dm_i = \sum_{i=1}^dm'_i = M$, we have $\sum_{i=1}^dm_iE_i\neq\sum_{i=1}^dm'_iE_i$.
Hence, energy levels cannot be integer multiples of each other, and energy degeneracy is also forbidden.
As an application of Theorem~\ref{Result:Upper-Bound} (and Corollary~\ref{Coro:CCUpper-Bound}), in Appendix~\ref{App:Proof-Coro:LinkCoro} we show the following bound [we implicitly assume the system Hamiltonian is the one realizing the given thermal state $\gamma$ with some temperature, and its energy eigenbasis defines the coherence, $R = {\rm Coh}$; also, we use $p_{\rm min}(\gamma)$ to denote the smallest eigenvalue of $\gamma$]:
\begin{theorem}\label{Coro:LinkCoro}
Given $0\le\epsilon,\delta<1\;\&\;0<\kappa<1$ and a full-rank thermal state $\gamma$.
Assume the system Hamiltonian satisfies the energy subspace condition.
Then for a Gibbs-preserving map $\mathcal{N}$ of $\gamma$ that is also coherence non-generating, we have
\begin{align}
C_{\rm (1)}^\epsilon(\mathcal{N})\le &P_{D_{\rm max}|{\rm Coh}}(\mathcal{N}) + \log_2\left(\mathcal{B}^{\delta}_\gamma\left({\Lambda^\mathcal{N}}\right) + \frac{2\sqrt{\delta}}{p_{\rm min}(\gamma)}+1\right)\nonumber\\
& + \log_2\frac{2^\kappa}{1-\epsilon}
\end{align}
for every $\Lambda^\mathcal{N}\in\mathcal{O}_{\rm Coh}^N(\kappa,\mathcal{N})$.
\end{theorem}
Theorem~\ref{Coro:LinkCoro} illustrates how a dynamical resource theory can bridge a pure thermodynamic quantity to a pure communication quantity.
To illustrate this, let us first consider the special case when $\mathcal{N}$ is {\em coherence-annihilating}~\cite{footnote:Coherence-Annihilating}.
In this case, one is able to choose $\mathcal{N} = \Lambda^\mathcal{N}$ and $\kappa = 0$.
Within this setup, if $\mathcal{N}$ can communicate a high amount of classical information, it necessarily requires a large bath to thermalize all its outputs.
On the other hand, if this channel has a small thermalization bath size, it unavoidably has a weak ability to communicate classical information.
Importantly, Theorem~\ref{Coro:LinkCoro} provides a physical message that is in spirit similar to the Landauer's principle~\cite{Landauer1961}.
Landauer's principle says that preparing a pre-defined pure state, e.g., $\ket{0}$, from an initial state $\rho$ requires at least $S(\rho)k_BT\ln2$ amount of energy~\cite{del_Lio2011}, where $S(\rho)\coloneqq-{\rm tr}(\rho\log_2\rho)$ is the von Neumann entropy of $\rho$.
In this sense, energy can be regarded as the thermodynamic cost needed to {\em create} classical information carried by an orthonormal basis $\{\ket{m}\}_{m=0}^{M-1}$.
Theorem~\ref{Coro:LinkCoro} provides a different, {\em dynamical} perspective: Under the given setting, transmitting $n$ bits of classical information, i.e., $C_{\rm (1)}^\epsilon(\mathcal{N}) = n$, necessarily requires the bath size $\mathcal{B}^\epsilon_\gamma(\mathcal{N})$ to be at least $2^n - 1$, up to some small terms.
In this case, the bath size can be interpreted as the thermodynamic cost needed to {\em transmit} classical information.
When the ability of the channel $\mathcal{N}$ to preserve coherence is turned on, interestingly, the cost to transmit classical information becomes a hybrid term: It is the bath size of $\mathcal{N}$'s incoherent version [i.e., $\Lambda^\mathcal{N}\in\mathcal{O}_{\rm Coh}^N(\kappa,\mathcal{N})$] {\em plus} the ability of $\mathcal{N}$ to preserve coherence.
In other words, this is the sum of the thermodynamic cost of $\mathcal{N}$'s classical counterpart, and the quantum effect maintained by $\mathcal{N}$.
By treating Landauer's principle as a bridge between thermodynamics and a {\em static} property of classical information, Theorem~\ref{Coro:LinkCoro} brings a connection between thermodynamics and a {\em dynamical} feature of classical information.

Note that, as expected, a full thermalization process (i.e., a state preparation channel of the given thermal state) cannot transmit any amount of classical information, since the bath size needed for thermalization is precisely zero.
This also implies that the inequality in Theorem~\ref{Coro:LinkCoro} is tight.
However, locally-performed thermalization processes can actually allow global transmission of classical information, which only needs shared randomness as a resource.
In this sense, thermalization together with a classical resource can still achieve nontrivial classical communication.
We will introduce this result more in-dept in Sec.~\ref{Sec:CCThermalization}, and the following section is for a tool needed for its proof.

\section{Application: Maintaining Orthogonal Maximal Entanglement}\label{Sec:MaintainingMaxEnt}
Once a question can be formulated into a classical communication problem, our approach can be used to study connections between the given question and different resource constraints.
To illustrate this, we study the following question: 
{\em How robust is the structure of orthogonal maximal entanglement under dynamics?}
As the motivation, a maximally entangled basis is a well-known tool in quantum information science, promising applications such as quantum teleportation~\cite{Bennett1993} and superdense coding~\cite{QCI-book}. 
The key is the simultaneous existence of maximal entanglement and orthogonality, and maintaining both of them through a physical evolution is vital for applications afterward. 
To model this question, we impose two restrictions in the classical communication scenarios used in this work: (i) The encoding $\{\rho_m\}_{m=0}^{M-1}$ are mutually orthonormal maximally entangled states  $\{\ket{\Phi_m}\}_{m=0}^{M-1}$.
(ii) The decoding $\{E_m\}_{m=0}^{M-1}$ are projective measurements done by a (sub-)basis of orthogonal maximally entangled states $\{\proj{\Phi'_m}\}_{m=0}^{M-1}$.
The corresponding one-shot classical capacity characterizes the ability of a given channel $\mathcal{N}$~\cite{footnote:Local-Dimension} to simultaneously maintain orthogonality and maximal entanglement:
\begin{align}\label{Eq:Capacity-Like-Measure}
C_{{\rm ME},(1)}^\epsilon(\mathcal{N})\coloneqq\log_2\max\{M\,|\,p_{s|{\rm ME}}(M,\mathcal{N})\ge1-\epsilon\},
\end{align}
where the success probability reads
$
p_{s|{\rm ME}}(M,\mathcal{N})\coloneqq
\sup_{\{\ket{\Phi_m}\},\{\ket{\Phi'_m}\}}\frac{1}{M}\sum_{m=0}^{M-1}\bra{\Phi'_m}\mathcal{N}(\proj{\Phi_m})\ket{\Phi'_m},
$
and the maximization is taken over all sets of orthogonal maximally entangled states of size $M$, denoted by $\{\ket{\Phi_m}\},\{\ket{\Phi'_m}\}$.
Thus, $C_{{\rm ME},(1)}^\epsilon(\mathcal{N})$ is the highest maintainable pairs of mutually orthonormal maximally entangled states under the dynamics $\mathcal{N}$, up to an error smaller than $\epsilon$.
To introduce the result, we say a state $\rho$ is {\em multi-copy nonlocal/steerable}~\cite{Palazuelos2012,Cavalcanti2013,Hsieh2016,Quintino2016} if there exists an integer $k$ such that $\rho^{\otimes k}$ is nonlocal/steerable.
Also, recall that $\mathbb{F}_R$ is the set of free operation of $R$-preservability defined in Sec.~\ref{Sec:Formulation}.
Then in Appendix~\ref{App:Proof-Result:ConverseBound} we show that
\begin{theorem}\label{Result:ConverseBound}
For a given $\mathcal{N}\in\mathcal{O}_R$ and $0\le\epsilon,\delta<1$ satisfying $\epsilon+\delta<1$, we have
\begin{align}
\alpha\times\sup_{\mathfrak{F}\in\mathbb{F}_R}C^\epsilon_{{\rm ME},(1)}\left[\mathfrak{F}(\mathcal{N})\right]\le P^\delta_{D_{\rm max}}(\mathcal{N}) + \log_2\frac{1}{1-\epsilon-\delta}
\end{align}
with $\alpha = 1$ when $R=$ athermality; $\alpha = \frac{1}{2}$ when $R=$ entanglement, free entanglement~\cite{Horodecki1998},
multi-copy nonlocality, and multi-copy steerability.
\end{theorem}
Theorem~\ref{Result:ConverseBound} provides upper bounds on $C_{\rm ME, (1)}(\mathcal{N})$ when $\mathcal{N}$ is a free operation of specific resources, which holds even when the channel is assisted by additional structures constrained by the resource  (which is given by $\mathbb{F}_R$).
Theorem~\ref{Result:ConverseBound} also brings an alternative operational interpretation of $R$-preservability: For the resources $R$ mentioned above, $R$-preservability bounds the channel's simultaneous maintainability of orthogonality and maximal entanglement in resource-constrained scenarios.

We remark that one can also interpret Eq.~\eqref{Eq:Capacity-Like-Measure} as a measure of the ability to admit superdense coding, and Theorem~\ref{Result:ConverseBound} therefore serves as an upper bound on this ability.
Furthermore, Theorem~\ref{Result:ConverseBound} brings a connection between fully entangled fraction~\cite{Horodecki1999-2,Albeverio2002} and $R$-preservability.
We leave the details in Appendix~\ref{App:Maintainability-MoreResults}.

\section{Application: Transmitting Information Through Thermalization}\label{Sec:CCThermalization}
The connection between Communication and Thermodynamics given by Theorem~\ref{Coro:LinkCoro} implies that it is impossible for a {\em full thermalization}, i.e., $\Phi_\gamma:(\cdot)\mapsto\gamma{\rm tr}(\cdot)$ for a given thermal state $\gamma$, to transmit any amount of classical information.
This is consistent with our physical intuition since a full thermalization means a complete destroy of the input information, leaving no freedom for the output.
However, this impossibility can be flipped with the help of pre-shared randomness, which is a pure classical resource.
More precisely, when local baths in a bipartite setting are allowed to be correlated through shared randomness, it is possible for locally performed full thermalizations to globally transmit classical information.

To make the statement precise, let us recall the following notions from Ref.~\cite{Hsieh2020-2}.
A bipartite channel $\mE_{\rm AB}$ on ${\rm AB}$ is called a {\em local thermalization} to a given pair of thermal states $(\gamma_{\rm A},\gamma_{\rm B})$ if 
(1) there exist an ancillary system ${\rm A'B'}$, unitary channels $\mathcal{U}_{\rm AA'},\mathcal{U}_{\rm BB'}$, and a thermal state $\neweta_{\rm A'B'}$ separable in ${\rm AB}$ bipartition such that
$
\mE_{\rm AB}(\rho) = {\rm tr}_{\rm A'B'}\left[\mathcal{U}_{\rm AA'}\otimes\mathcal{U}_{\rm BB'}(\rho\otimes\neweta_{\rm A'B'})\right],
$
and (2) ${\rm tr}_{\rm B}\circ\mE_{\rm AB}(\rho) = \gamma_{\rm A}$ and ${\rm tr}_{\rm A}\circ\mE_{\rm AB}(\rho) = \gamma_{\rm B}$ for all $\rho$.
As mentioned in Ref.~\cite{Hsieh2020-2}, local thermalization is {\em local} in two senses: It is a local operation plus pre-shared randomness channel, and it is locally equivalent to a full thermalization process.
The main message from Ref.~\cite{Hsieh2020-2} is that {\em entanglement preserving local thermalization} exists; that is, for every non-pure full-rank $\gamma_{\rm A},\gamma_{\rm B}$, there exists a local thermalization $\mE_{\rm AB}$ such that $\mE_{\rm AB}(\rho)$ is entangled for some $\rho$.
It turns out that such channels can also transmit classical information.
To state the result, consider a tripartite system ${\rm ABC}$ with local dimension $d,d,d^2+1$, respectively.
We focus on the bipartition ${\rm A|BC}$; namely, the subsystem ${\rm C}$ can be understood as an ancillary system possessed by the agent in ${\rm B}$.
Let $\gamma_{{\rm X}}$ be the thermal state of the subsystem ${\rm X=A,B,C}$ with temperature $T_{\rm X}$ and Hamiltonian $H_{\rm X}$.
Following Ref.~\cite{Hsieh2020-2}, we assume each $H_{\rm X}$ is non-degenerate and finite-energy.
Then, Combining Ref.~\cite{Hsieh2020-2} and results in Sec.~\ref{Sec:MaintainingMaxEnt}, in Appendix~\ref{App:Proof-Result:GeneralizedEPLT} we prove that locally performed thermalization is able to transmit classical information with the help of pre-shared randomness [let $p_{\rm min|AB}\coloneqq\min\{p_{\rm min}(\gamma_{{\rm A}});p_{\rm min}(\gamma_{{\rm B}})\}$ and $\gamma_{\rm BC}\coloneqq\gamma_{\rm B}\otimes\gamma_{\rm C}$; see also Fig.~\ref{Fig:EPLT-CC}]: 
\begin{theorem}\label{Result:GeneralizedEPLT}
When $d<\infty$ and $T_{{\rm X}}>0$, there exists an entanglement preserving local thermalization to $(\gamma_{\rm A},\gamma_{\rm BC})$, denoted by $\mE_{\rm A|BC}$, such that
\begin{align}
C_{\rm (1)}^\epsilon(\mE_{\rm A|BC})\ge\log_2d^2
\end{align}
for every $\epsilon\ge\left(1 - \frac{1}{d^2}\right)\left(1-dp_{\rm min|AB}\right)$.
\end{theorem}
Under the global dynamics given by Theorem~\ref{Result:GeneralizedEPLT}, although the local agents will observe a full thermalization process, transmitting classical information is still possible via the joint bipartite dynamics.
This is because classical information can be locally encoded in the global quantum correlation.
Note that the amount of transmissible classical information largely depends on local temperatures: When the local systems are too cold (which corresponds to the limit $p_{\rm min|AB}\to0$), a low temperature will force the local thermal state be closed to a pure state (here we assume no energy degeneracy), resulting in no possibility to maintain global correlation.
Also note that although $\mE_{\rm A|BC}$ is locally identical to a full thermalization process, its global behavior is far from a thermalization.
This gap between local and global dynamics leads to the possibility for a classical communication through local thermalizations when shared randomness is accessed to.
Notably, this illustrate that, although being a classical resource, shared randomness provides advantages to maintain quantum correlation and assist classical communication through thermalization processes.
\begin{figure}
\scalebox{0.8}{\includegraphics{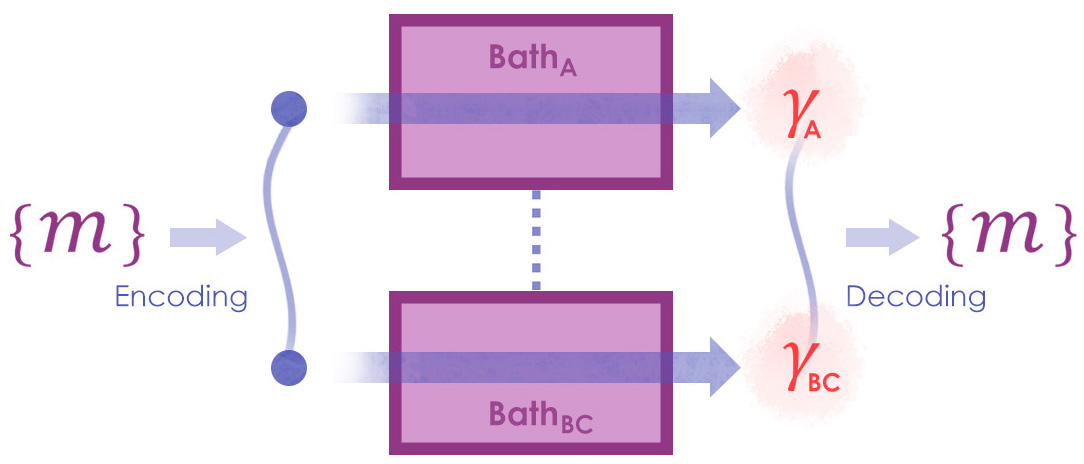} }
\caption{
Schematically, Theorem~\ref{Result:GeneralizedEPLT} implies that when local baths are classically correlated (the dashed line), it is possible to transmit classical information through locally performed full thermalization processes.
}
\label{Fig:EPLT-CC} 
\end{figure}

\section{Conclusion}\label{Sec:Conclusion}
We study classical communication scenarios with free operations of a given resource as the information processor.
The one-shot classical capacity is upper bounded by resource preservability~\cite{Hsieh2020-1} plus a term of resourceless contribution.
This upper bound provides an alternative interpretation of resource preservability. 
Furthermore, when asymmetry is the resource, a lower bound can also be obtained.

As an application, we use our result to bridge two seemingly different concepts: Under the thermalization model given by Ref.~\cite{Sparaciari2019}, for every Gibbs-preserving coherence non-generating channel, its smallest {\em channel bath size} (i.e., a bath size needed to thermalize all outputs of a given channel) plus its coherence preservability will upper bound its one-shot classical capacity.
Thus, under this setting, transmitting $n$ bits of classical information through a coherence-annihilating channel requires the channel bath size to be at least $2^n-1$, up to some small terms.
This reveals an implicit thermodynamic cost of transmitting classical information, providing a dynamical analogue of the Landauer's principle~\cite{Landauer1961} and illustrating how a dynamical resource theory allows applications to connect different dynamical phenomena.

We further apply our approach to study channel's simultaneous maintainability of orthogonality and maximal entanglements.
Formulating the question into a communication form, a capacity-like measure can be introduced and upper bounded by resource preservability.

Finally, applying the thermalization channel introduced in Ref.~\cite{Hsieh2020-2} to a bipartite setting, we show that classically correlated local baths allow a decent amount of one-shot classical capacity even when both local systems are completely thermalized.
Hence, classical information processing and thermalization processes can coexist, which only requires shared randomness as a resource.
This result also means that when a many-body system is in contact with a global bath having classical correlations within, it is possible to maintain classical information even after it has been thermalized locally.
Share randomness as a resource is enough to guarantee that certain amounts of classical information can be extracted after local thermalizations, provided that local temperatures are not too low.

Several open questions remain.
First, currently all the results related to resource preservability (e.g., Ref.~\cite{Hsieh2020-1}) are focusing on the one-shot regime, and a future direction is to understand its asymptotic behavior.
Furthermore, whether one can derive a lower bound similar to Theorem~\ref{Result:Asymmetry-Lower-Bounds} in terms of resource preservability and even extend the result to other state resources are still unknown.
These questions could be difficult and largely depend on the choice of resources, since, e.g., the lower bound on the capacity used in Ref.~\cite{Korzekwa2019} is given by the information spectrum relative entropy, which is not a contractive generalized distance measure~\cite{Leditzky-PhD} and hence cannot induce legal resource preservability monotone.
Also, whether one can obtain any result similar to Theorem~\ref{Result:ConverseBound} in the context of coherence is still unknown.
Finally, as a direct consequence of Theorem~\ref{Coro:LinkCoro}, we have the following conjecture:
{\em
Transmitting $n$ bits of classical information through a Gibbs-preserving coherence-annihilating channel requires the corresponding channel bath size to be at least $2^n-1$.
}
This conjecture could largely depend on the underlying communication setup and the thermalization model.

We hope the physical messages provided by this work can offer alternative interpretations in the interplay of dynamical resource theory, classical communication, thermodynamics, and different forms of inseparability.

\section*{Acknowledgements}
We thank Antonio Ac\'in, Stefan B$\ddot{\rm a}$uml, Dario De Santis, Yeong-Cherng Liang, Matteo Lostaglio, Mohammad Mehboudi, Gabriel Senno, and Ryuji Takagi for fruitful discussions and comments.
This project is part of the ICFOstepstone - PhD Programme for Early-Stage Researchers in Photonics, funded by the Marie Sk\l odowska-Curie Co-funding of regional, national and international programmes (GA665884) of the European Commission, as well as by the ‘Severo Ochoa 2016-2019' program at ICFO (SEV-2015-0522), funded by the Spanish Ministry of Economy, Industry, and Competitiveness (MINECO).
We also acknowledge support from the Spanish MINECO (Severo Ochoa SEV-2015-0522), Fundaci\'o Cellex and Mir-Puig, Generalitat de Catalunya (SGR1381 and CERCA Programme).

\onecolumngrid

\appendix

\section{Being Realizable Without Consuming Resource and Being Resource Non-Generating Are Not Equivalent}\label{App:NotionCostless}
In order to study channels constrained by a given static resource $R$, it is straightforward to expect these channels to be {\em free from $R$}. 
In this work, we depict this by requiring those channels to be ``unable to generate $R$.'' 
This coincides with the notion of being {\em $R$ non-generating}~\cite{RT-RMP}, and hence being free operations of $R$.
Meanwhile, there is another possible approach, which is requiring those channels to be ``realizable without consuming $R$.'' 
It turns out that this latter option is less generic and accessible compared with the former. 
This Appendix aims to briefly make this difference clear.

In a given resource theory, either static, dynamical, or a more general one, requiring an operation to be free from the resource sometimes includes two seemingly equivalent concepts implicitly; namely, being realizable without consuming the resource, and being unable to generate the resource.
While these two concepts match for some settings, in general they are not equivalent.
For instance, the resource theories of entanglement equipped with {\em local operation and classical communication} (LOCC) channels or {\em local operation plus pre-shared randomness} (LOSR) channels allow this property, and so does the resource theory of nonlocality with LOSR channels.
This is because LOCC (so do LOSR) channels can be realized without touching and using any entangled state.
Nevertheless, the resource theory of athermality demonstrates a counterexample.
In this case, the only free state is the state in thermal equilibrium, i.e., the thermal state $\gamma$.
Physically, it is impossible to realize any non-trivial channel only within thermal equilibrium (the only realizable one is the state preparation channel of $\gamma$, since one can artificially switch $\gamma$ with the input and discard the original system).
On the other hand, a commonly used free operation is the thermal operation, which takes the form $(\cdot)_{\rm A}\mapsto{\rm tr}_{\rm B}[U_{\rm AB}((\cdot)_{\rm A}\otimes\gamma_{\rm B})U_{\rm AB}^\dagger]$, where $U_{\rm AB}$ conserves the total energy (i.e., it commutes with the total Hamiltonian).
One can see that any non-trivial thermal operation (which is athermality non-generating) needs a non-trivial unitary $U_{\rm AB}$, and hence includes effects out of thermal equilibrium;
i.e., it is not realizable without consuming athermality.

Apart from resource theories of states, there are also instances in dynamical resource theories illustrating this difference.
In the resource theory for non-signaling assisted classical communication~\cite{Takagi2019-3}, the free quantities are state preparation channels, and it is impossible to output channels useful for classical communication if one only uses state preparation channels to implement free super-channels.
Similarly, in the theory of resource preservability~\cite{Hsieh2020-1}, it is again impossible to output resourceful channels when one only uses resource annihilating channels to implement free super-channels.

If one upgrades the discussion to general and abstract considerations, it can be tough to access the detailed physical structures of operations.
Consequently, the best one can do is to analyze an operation by comparing its inputs and outputs. 
This is also the only way to check whether an operation is free from the given resource.
Hence, ``zero ability to generate the resource'' ends up to be the most feasible and well-defined way to depict ``begin free'' in the most general extend when further structures and contexts are not available.
Being realizable without consuming the resource is an additional property that can be satisfied in certain cases, but this notion could be generally ill-defined.

Note that this is also why in a general, model-independent level, the definition of being resource non-generating only requires no generation of the resource for free inputs: Before introducing free operations, we cannot compare and order different resourceful states, and the only existing concept before defining free operations is ``whether the quantity is resourceful or not."
This gives us the most extensive range to clarify the notion of ``being free from the resource.'' 
It also briefly summarizes the features of central ingredients in a resource theory: Free states give us detection, free operations give us comparison, and monotones give us quantification.

Due to the above discussion, in this work, we depict a channel as constrained by a resource if it is a free operation.

\section{Assumptions on State Resource Theories for Resource Preservability Theories}\label{App:Basic-Assumptions}
To have a general study that is also analytically feasible, we need to impose certain assumptions on the state resource theories considered in this work.
Let $(R,\mathcal{F}_R,\mathcal{O}_R)$ be a given state resource theory.
Then we consider
\begin{enumerate}
\item\label{Assumption:Identity} Identity channel and partial trace are both free operations; namely, they are both in $\mathcal{O}_R$.
\item\label{Assumption:Convex} Free operations are closed under tensor products, convex sums, and compositions: For every $\mE,\mE'\in\mathcal{O}_R$ and $p\in[0,1]$, we have $\mE\otimes\mE'\in\mathcal{O}_R$, $p\mE + (1-p)\mE'\in\mathcal{O}_R$, and $\mE\circ\mE'\in\mathcal{O}_R$.
\item\label{Assumption:AbsolutelyFree} For every system ${\rm S'}$ there exists a state $\sigma_{\rm S'}$ such that $(\cdot)_{\rm S}\mapsto(\cdot)_{\rm S}\otimes\sigma_{\rm S'}$ is a free operation.
\end{enumerate}
Assumptions~\ref{Assumption:Identity} and~\ref{Assumption:Convex} are always assumed in this work in order to capture the necessary properties of a monotone, and we leave Assumption~\ref{Assumption:AbsolutelyFree} optional.
This is slightly different from Ref.~\cite{Hsieh2020-1}, and our motivation is to relax the assumptions made by Ref.~\cite{Hsieh2020-1} to achieve a general consideration admitting more applicable cases.
We briefly explain each assumptions.
Assumption~\ref{Assumption:Identity} follows from our conceptual expectation; that is, ``doing nothing'' and ``ignoring part of the system'' are both unable to generate $R$.
Assumption~\ref{Assumption:Convex} implies that if two channels are unable to generate $R$, then neither can their simultaneous application (tensor product), classical mixture (convex sum), and sequential application (composition). 
Finally, Assumption~\ref{Assumption:AbsolutelyFree} ensures that there always exists a ``free extension,'' which automatically implies the state $\sigma_{\rm S'}$ is free an hence $\mathcal{F}_R\neq\emptyset$ (to see this, consider ${\rm tr}_{\rm S}$ and use Assumptions~\ref{Assumption:Identity} and~\ref{Assumption:Convex}).
Note that Assumption~\ref{Assumption:AbsolutelyFree} is only imposed on systems with proper system sizes.
For example, in the resource theory of entanglement, steering, and nonlocality, ${\rm S'}$ must be bipartite (and we always assume equal local dimension); in the resource theory of athermality, ${\rm S'}$ can only have dimension $d^k$ with some positive integer $k$, where $d$ is the dimension of the given thermal state.

Many known resource theories share these assumptions.
For instance, Assumptions~\ref{Assumption:Identity},~\ref{Assumption:Convex}, and~\ref{Assumption:AbsolutelyFree} are satisfied by the sets of LOCC channels, LOSR channels, Gibbs-preserving maps, and $G$-covariant channels (in multipartite cases, we consider the group $G^{\otimes k}\coloneqq\{\bigotimes_{i=1}^kU_i\,|\,U_i\in G\;\forall i\}$).
Note that Assumption~\ref{Assumption:AbsolutelyFree} holds since for every system ${\rm S'}$ with dimension $d_{\rm S'}$ the mapping $(\cdot)\mapsto(\cdot)\otimes\frac{\id_{\rm S'}}{d_{\rm S'}}$ is an LOSR and $G$-covariant channel.
The case of Gibbs-preserving maps (with the thermal state $\gamma$) follows from the fact that $(\cdot)\mapsto(\cdot)\otimes\gamma^{\otimes k}$ is Gibbs-preserving for all $k$.
This implies the validity of Assumptions~\ref{Assumption:Identity},~\ref{Assumption:Convex}, and~\ref{Assumption:AbsolutelyFree} in the following state resource theories, which covers most of the cases studied in this work:
(i) entanglement and free entanglement~\cite{Horodecki1998} equipped with LOCC or LOSR channels, (ii) nonlocality, steering, multi-copy nonlocality, and multi-copy steering equipped with LOSR channels (see Appendix~\ref{App:LOSR-Nonlocality-Steering} for a discussion), (iii) $G$-asymmetry equipped with $G$-covariant channels, and (iv) athermality equipped with Gibbs-preserving maps.

It turns out that, by using Assumptions~\ref{Assumption:Identity},~\ref{Assumption:Convex}, and~\ref{Assumption:AbsolutelyFree}, we can prove a generalized version of Theorem 2 in Ref.~\cite{Hsieh2020-1}, which is summarized as follows:
\begin{atheorem}\label{Thm:GeneralizedThm2}{\em\cite{Hsieh2020-1}}
$(R,\mathcal{F}_R,\mathcal{O}_R)$ is a state resource theory satisfying Assumptions~\ref{Postulate:Identity} and~\ref{Postulate:Composition}.
$D$ is a contractive generalized distance measure of states.
Then $P_D$ satisfies
\begin{enumerate}
\item $P_D(\mathcal{N})\ge0$ and $P_D(\mathcal{N}) = 0$ if $\mathcal{N}\in\mathcal{O}_R^N$.
\item $P_D[\mathfrak{F}(\mE)]\le P_D(\mE)$ for every channel $\mE$ and free super-channel $\mathfrak{F}\in\mathbb{F}_R$.
\end{enumerate}
If Assumption~\ref{Assumption:AbsolutelyFree} holds, then we have
\begin{align}
P_D(\mathcal{N}\otimes\mathcal{N}')\ge P_D(\mathcal{N})
\end{align}
for every $\mathcal{N},\mathcal{N}'\in\mathcal{O}_R$, and the equality holds if $\mathcal{N}'\in\wt{\mathcal{O}}_R^N$.
\end{atheorem}
\begin{proof}
Apart from Eq.~(49) in Ref.~\cite{Hsieh2020-1}, the proof is the same with the one of Theorem 2 in Ref.~\cite{Hsieh2020-1} (see Eqs.~(48) and~(50) in Ref.~\cite{Hsieh2020-1}).
Note that Eq.~(48) in Ref.~\cite{Hsieh2020-1} works for every channel, which explains the validity of statement 2 in this theorem.
It remains to show Eq.~(7) in Ref.~\cite{Hsieh2020-1}, which can be seen by the following alternative proof:
\begin{align}
P_D(\mE_{\rm S}\otimes\mE_{\rm S'})& = \inf_{\Lambda_{\rm SS'}\in\mathcal{O}_R^N}\sup_{\rm A}D\left[(\mE_{\rm S}\otimes\mE_{\rm S'}\otimes\wt{\Lambda}_{\rm A})(\rho_{\rm SS'A}),(\Lambda_{\rm SS'}\otimes\wt{\Lambda}_{\rm A})(\rho_{\rm SS'A})\right]\nonumber\\
&\ge \inf_{\Lambda_{\rm SS'}\in\mathcal{O}_R^N}\sup_{\rm A}D\left[(\mE_{\rm S}\otimes\wt{\Lambda}_{\rm A})(\rho_{\rm SA}),[({\rm tr}_{\rm S'}\circ\Lambda_{\rm SS'})\otimes\wt{\Lambda}_{\rm A}](\rho_{\rm SS'A})\right]\nonumber\\
&\ge\inf_{\Lambda_{\rm SS'}\in\mathcal{O}_R^N}\sup_{\rm A}D\left[(\mE_{\rm S}\otimes\wt{\Lambda}_{\rm A})(\rho_{\rm SA}),[({\rm tr}_{\rm S'}\circ\Lambda_{\rm SS'})\otimes\wt{\Lambda}_{\rm A}](\rho_{\rm SA}\otimes\sigma_{\rm S'})\right]\nonumber\\
&\ge\inf_{\Lambda_{\rm S}\in\mathcal{O}_R^N}\sup_{\rm A}D\left[(\mE_{\rm S}\otimes\wt{\Lambda}_{\rm A})(\rho_{\rm SA}),(\Lambda_{\rm S}\otimes\wt{\Lambda}_{\rm A})(\rho_{\rm SA})\right]\nonumber\\
& = P_D(\mE_{\rm S}).
\end{align}
The second line follows from the data-processing inequality and the fact that $\mE_{\rm S}\otimes\mE_{\rm S'}$ has a well-defined marginal in ${\rm S}$~\cite{Acknowledgement}.
In the third line, $\rho_{\rm SA}\otimes\sigma_{\rm S'}$ forms a sub-optimal range of the maximization, where $\sigma_{\rm S'}$ is the state guaranteed by Assumption~\ref{Assumption:AbsolutelyFree} that allows the map $(\cdot)\mapsto(\cdot)\otimes\sigma_{\rm S'}$ to be a free operation of $R$.
Together with Assumptions~\ref{Assumption:Identity} and~\ref{Assumption:Convex}, we learn that $({\rm tr}_{\rm S'}\circ\Lambda_{\rm SS'})[(\cdot)\otimes\sigma_{\rm S'}]\in\mathcal{O}_R^N$ is a resource annihilating channel, which forms a sub-optimal range of the minimization and implies the fourth line.
Hence, Assumptions~\ref{Postulate:Identity},~\ref{Postulate:Composition}, and~\ref{Assumption:AbsolutelyFree} are enough to ensure the correcteness of Theorem 2 in Ref.~\cite{Hsieh2020-1}.
\end{proof}
Theorem~\ref{Thm:GeneralizedThm2} generalizes Theorem 2 in Ref.~\cite{Hsieh2020-1} by relaxing the assumptions of absolutely free states (i.e., the assumptions (R1) and (R3) in Ref.~\cite{Hsieh2020-1}) into Assumption~\ref{Assumption:AbsolutelyFree}.
Furthermore, there is no need to assume the convexity of $\mathcal{F}_R$. 
Another remark is that non-increasing under free super-channel actually works for every channel, including channels that are not free operations.
This is a useful observation when one needs to consider the smooth version of $R$-preservability, e.g., in the next sub-section.

\subsection{Properties of $P_D^\delta$}\label{App:Proof-Eq:Corollary}
We remark that $P_D^\delta$, which can be interpreted as the smooth version of $P_D$, still possesses the expected properties that a monotone should have.
First, if $\mathcal{N}\in\mathcal{O}_R^N$, then we have $P_D^\delta(\mathcal{N})\coloneqq\inf_{\norm{\mathcal{N} - \mathcal{N}'}_\diamond\le2\delta}P_{D}(\mathcal{N}') \le P_D(\mathcal{N}) = 0$.
The non-increasing property under free super-channels can be summarized in the following lemma:
\begin{alemma}\label{Lemma:SmoothNonIncreasing}
For every channel $\mathcal{N}$, $0\le\delta\le1$, and $\mathfrak{F}\in\mathbb{F}_R$, we have
\begin{align}
P_D^\delta[\mathfrak{F}(\mathcal{N})]\le P_D^\delta(\mathcal{N}).
\end{align}
\end{alemma}
\begin{proof}
We note the following estimate first:
\begin{align}\label{Eq:EstimateAncillaryLambda}
\norm{(\mathcal{N}-\mathcal{N}')\otimes\wt{\Lambda}}_\diamond&\le\norm{(\mathcal{N}-\mathcal{N}')\otimes\mathcal{I}}_\diamond\nonumber\\
&\le\norm{\mathcal{N}-\mathcal{N}'}_\diamond.
\end{align}
The first inequality follows from the data processing inequality, or equivalently, the contractivity of the trace norm; the second ineuqality follows from the definition of the diamond norm.
Recall that $\mathfrak{F}\in\mathbb{F}_R$ will take the form $\mathfrak{F}(\mathcal{N}) = \Lambda_+\circ(\mathcal{N}\otimes\wt{\Lambda})\circ\Lambda_-$. 
We conclude that
\begin{align}\label{Eq:ComputationMathfrak}
\norm{\mathfrak{F}(\mathcal{N}) - \mathfrak{F}(\mathcal{N}')}_\diamond& = \norm{\Lambda_+\circ\left[(\mathcal{N} - \mathcal{N}')\otimes\wt{\Lambda}\right]\circ\Lambda_-}_\diamond\nonumber\\
&\le\norm{\left[(\mathcal{N} - \mathcal{N}')\otimes\wt{\Lambda}\right]\circ\Lambda_-}_\diamond\nonumber\\
&\coloneqq\sup_{{\rm A};\rho_{\rm SA}}\norm{\left[(\mathcal{N} - \mathcal{N}')\otimes\wt{\Lambda}\otimes\mathcal{I}_{\rm A}\right]\circ(\Lambda_-\otimes\mathcal{I}_{\rm A})(\rho_{\rm SA})}_1\nonumber\\
&\le\norm{(\mathcal{N} - \mathcal{N}')\otimes\wt{\Lambda}}_\diamond\nonumber\\
&\le\norm{\mathcal{N} - \mathcal{N}'}_\diamond,
\end{align}
where the second line follows from data-processing inequality, the fourth line is because $(\Lambda_-\otimes\mathcal{I}_{\rm A})(\rho_{\rm SA})$ induces a sub-optimal range of the maximization in the definition of diamond norm.
In the last line we use Eq.~\eqref{Eq:EstimateAncillaryLambda}.
Now, direct computation shows
\begin{align}
P_D^\delta(\mathcal{N})&\coloneqq\inf_{\norm{\mathcal{N} - \mathcal{N}'}_\diamond\le2\delta}P_{D}(\mathcal{N}')\nonumber\\
&\ge\inf_{\norm{\mathfrak{F}(\mathcal{N}) - \mathfrak{F}(\mathcal{N}')}_\diamond\le2\delta}P_{D}(\mathcal{N}')\nonumber\\
&\ge\inf_{\norm{\mathfrak{F}(\mathcal{N}) - \mathfrak{F}(\mathcal{N}')}_\diamond\le2\delta}P_{D}[\mathfrak{F}(\mathcal{N}')]\nonumber\\
&\ge\inf_{\norm{\mathfrak{F}(\mathcal{N}) - \mathcal{N}''}_\diamond\le2\delta}P_{D}(\mathcal{N}'')\nonumber\\
& = P_D^\delta[\mathfrak{F}(\mathcal{N})].
\end{align}
From Eq.~\eqref{Eq:ComputationMathfrak} we learn that all channels $\mathcal{N}'$ satisfying $\norm{\mathcal{N} - \mathcal{N}'}_\diamond\le2\delta$ form a subset of all channels $\mathcal{N}'$ satisfying $\norm{\mathfrak{F}(\mathcal{N}) - \mathfrak{F}(\mathcal{N}')}_\diamond\le2\delta$.
This explains the second line.
The third line follows from Theorem~\ref{Thm:GeneralizedThm2} (note that $\mathcal{N}'$ could be outside $\mathcal{O}_R$, but it is still a channel).
In the fourth line, we have the set of all channels of the form $\mathfrak{F}(\mathcal{N}')$ be a subset of the set of all channels.
This shows the desired claim.
\end{proof}

\subsection{LOSR Channels as Free Operations of Nonlocality, Steering, Multi-Copy Nonlocality, and Multi-Copy Steering}\label{App:LOSR-Nonlocality-Steering}
Appendix A.1 in Ref.~\cite{Hsieh2020-1} explains that LOSR channels can be free operations of nonlocality.
Here, we briefly show that LOSR channels can also be free operations of three other different forms of inseparabilities: Steering~\cite{Wiseman2007,Jones2007,steering-review,RMP-steering}, multi-copy nonlocality~\cite{Palazuelos2012,Cavalcanti2013}, and multi-copy steering~\cite{Hsieh2016,Quintino2016}.
Formally, a LOSR channel in a given bipartite system ${\rm AB}$ is given by the following form:
\begin{align}\label{Eq:LOSR}
\mE\coloneqq\int(\mE^{\rm A}_\lambda\otimes\mE^{\rm B}_\lambda)p_\lambda d\lambda,
\end{align}
where the integration is taken over the parameter $\lambda$.
Physically, it is a convex mixture of local dynamics.
Now, in the given bipartite system ${\rm AB}$, a state is {\em unsteerable from ${\rm A}$ to ${\rm B}$}~\cite{Wiseman2007,Jones2007,steering-review,RMP-steering}, or simply ${\rm A\to B}$ {\em unsteerable}, if for every local POVMs $\{E_{a|x}^{\rm A}\}$ in ${\rm A}$ and $\{E_{b|y}^{\rm B}\}$ in ${\rm B}$, one can write
\begin{align}\label{Eq:LHV}
{\rm tr}\left[\left(E_{a|x}^{\rm A}\otimes E_{b|y}^{\rm B}\right)\rho\right] = \int_{\lambda\in\Lambda_{\rm LHS}}P(a|x,\lambda){\rm tr}\left(E_{b|y}^{\rm B}\sigma_\lambda\right)p_\lambda d\lambda
\end{align}
for some variable $\lambda$ in a set $\Lambda_{\rm LHS}$, some probability distributions $P(a|x,\lambda),p_\lambda$, and some local states $\sigma_\lambda$ in ${\rm B}$.
In other words, a state is ${\rm A\to B}$ unsteerable if every outcome of local measurements in ${\rm B}$ is indistinguishable from the outputs of pre-shared randomness combined with local quantum theory in ${\rm B}$.
Such models are called {\em local hidden state models}~\cite{Wiseman2007,Jones2007,steering-review,RMP-steering}, as depicted by $\Lambda_{\rm LHS}$.
States that are not ${\rm A\to B}$ unsteerable are said to be ${\rm A\to B}$ {\em steerable}.

To see why LOSR channels can be free operations for steering, consider an LOSR channel $\mE\coloneqq\int_\nu \left(\mE_\nu^{\rm A}\otimes\mE_\nu^{\rm B}\right)q_\nu d\nu$ and the following computation
\begin{align}\label{Eq:ComputationLHV}
{\rm tr}\left[\left(E_{a|x}^{\rm A}\otimes E_{b|y}^{\rm B}\right)\mE(\rho)\right] &= \int_\nu {\rm tr}\left[\left(E_{a|x}^{\rm A}\otimes E_{b|y}^{\rm B}\right)\left(\mE_\nu^{\rm A}\otimes\mE_\nu^{\rm B}\right)(\rho)\right]q_\nu d\nu\nonumber\\
& = \int_\nu {\rm tr}\left[\left(\mE_\nu^{\rm A,\dagger}\left(E_{a|x}^{\rm A}\right)\otimes \mE_\nu^{\rm B,\dagger}\left(E_{b|y}^{\rm B}\right)\right)\rho\right]q_\nu d\nu,
\end{align}
where $\mE_\nu^{\rm A,\dagger}\left(E_{a|x}^{\rm A}\right)$ and $\mE_\nu^{\rm B,\dagger}\left(E_{b|y}^{\rm B}\right)$ again form local POVMs since $\mE_\nu^{\rm A,\dagger},\mE_\nu^{\rm B,\dagger}$ are completely-positive unital maps.
Hence, when $\rho$ is ${\rm A\to B}$ unsteerable, it means, for every $\nu$, we can write ${\rm tr}\left[\left(\mE_\nu^{\rm A,\dagger}\left(E_{a|x}^{\rm A}\right)\otimes \mE_\nu^{\rm B,\dagger}\left(E_{b|y}^{\rm B}\right)\right)\rho\right]$ as Eq.~\eqref{Eq:LHV}.
This means the output of Eq.~\eqref{Eq:ComputationLHV} is again described by Eq.~\eqref{Eq:LHV}.

With the notions of nonlocality and steering, we say a state $\rho$ is {\em multi-copy nonlocal}~\cite{Palazuelos2012} (and, similarly, {\em multi-copy ${\rm A\to B}$ steerable}~\cite{Hsieh2016,Quintino2016}) if $\rho^{\otimes k}$ is nonlocal (${\rm A\to B}$ steerable) for some positive integer $k$.
One can see that LOSR channels again act as free operations for these two resources.
To see this, it suffices to observe that if $\mE$ is an LOSR channel in a given bipartition, then $\mE^{\otimes k}$ will again be an LOSR channel in the same bipartition.
More precisely, consider an LOSR channel $\mE_{\rm AB}$ in ${\rm AB}$ bipartition. 
Suppose $\rho$ is multi-copy local (${\rm A\to B}$ unsteerable) in this bipartition; namely, $\rho^{\otimes k}$ is local (${\rm A\to B}$ unsteerable) for all $k$. 
Then, for all $k$, $\left[\mE_{\rm AB}(\rho)\right]^{\otimes k} = \mE_{\rm AB}^{\otimes k}\left(\rho^{\otimes k}\right)$ must be local (${\rm A\to B}$ unsteerable) since $\rho^{\otimes k}$ is local (${\rm A\to B}$ unsteerable) and $\mE_{\rm AB}^{\otimes k}$ is again an LOSR channel in the ${\rm AB}$ bipartition.
This shows that LOSR channels can be free operations of multi-copy nonlocality and multi-copy steering.

\section{Proofs of Theorem~\ref{Result:Upper-Bound} and Corollary~\ref{Coro:CCUpper-Bound}}\label{App:Proof-Result:Upper-Bound}
First, we note the following lemma similar to Fact E.2 in Ref.~\cite{Hsieh2020-1}.
This will enable us to obtain an equivalent representation of $P_{D_{\rm max}}$ defined in Eq.~\eqref{Eq:P_D}.
In what follows, the maximization $\sup_{\wt{\Lambda}_{\rm A},\rho_{\rm SA}}$ is taken over all ancillary systems ${\rm A}$, absolutely resource annihilating channels $\wt{\Lambda}_{\rm A}$, and joint input states $\rho_{\rm SA}$.
Note that the maximization includes the trivial ancillary system (i.e.,  the one with dimension 1), which means it also covers the case when there is no ancillary system.
\begin{alemma}\label{Lemma:AlternativeRepresentation}
Given two channels $\mathcal{N}$ and $\mathcal{E}$, then we have
\begin{align}
\sup_{\wt{\Lambda}_{\rm A},\rho_{\rm SA}}\inf\left\{\lambda\ge0\,|\,0\le[(\lambda\mathcal{E}-\mathcal{N})\otimes\wt{\Lambda}_{\rm A}](\rho_{\rm SA})\right\} = \inf\left\{\lambda\ge0\,|\,0\le[(\lambda\mathcal{E}-\mathcal{N})\otimes\wt{\Lambda}_{\rm A}](\rho_{\rm SA})\;\forall{\rm A},\wt{\Lambda}_{\rm A},\rho_{\rm SA}\right\}.
\end{align}
\end{alemma}
\begin{proof}
Let $\mathcal{L}_{\bf A}\coloneqq\left\{\lambda\,|\,0\le[(\lambda\mathcal{E}-\mathcal{N})\otimes\wt{\Lambda}_{\rm A}](\rho_{\rm SA})\right\}$, where ${\bf A}\coloneqq({\rm A},\wt{\Lambda}_{\rm A},\rho_{\rm SA})$ is a specific combination of ${\rm A}$, $\wt{\Lambda}_{\rm A}$, and $\rho_{\rm SA}$.
Then the left-hand-side is $\sup_{\bf A}\inf\{\lambda\,|\,\lambda\in\mathcal{L}_{\bf A}\}$, and the right-hand-side is $\inf\left\{\lambda\,|\,\lambda\in\bigcap_{\bf A}\mathcal{L}_{\bf A}\right\}$.
The inequality ``$\le$'' follows since $\bigcap_{\bf A}\mathcal{L}_{\bf A}\subseteq\mathcal{L}_{\bf A'}$ for all ${\bf A'}$.
To show the opposite, consider an arbitrary positive integer $k$.
Then there exist ${\bf A}_k$  and $\lambda_k\in\mathcal{L}_{{\bf A}_k}$ such that
\begin{align}
&\inf\left\{\lambda\,|\,\lambda\in\mathcal{L}_{{\bf A}_k}\right\}\le \sup_{\bf A}\inf\{\lambda\,|\,\lambda\in\mathcal{L}_{\bf A}\}<\inf\left\{\lambda\,|\,\lambda\in\mathcal{L}_{{\bf A}_k}\right\} + \frac{1}{k};\\
&\lambda_k - \frac{1}{k}<\inf\left\{\lambda\,|\,\lambda\in\mathcal{L}_{{\bf A}_k}\right\}\le\lambda_k.
\end{align}
This means $\inf\{\lambda\,|\,\lambda\in\mathcal{L}_{\bf A}\}<\lambda_k + \frac{1}{k}$ for all ${\bf A}$, which further implies $\lambda_k + \frac{1}{k}\in\bigcap_{\bf A}\mathcal{L}_{\bf A}$. 
We conclude that
\begin{align}
\inf\left\{\lambda\,|\,\lambda\in\bigcap_{\bf A}\mathcal{L}_{\bf A}\right\}&\le\lambda_k + \frac{1}{k}\nonumber\\
&\le\inf\left\{\lambda\,|\,\lambda\in\mathcal{L}_{{\bf A}_k}\right\}+\frac{2}{k}\nonumber\\
&\le\sup_{\bf A}\inf\{\lambda\,|\,\lambda\in\mathcal{L}_{\bf A}\}+\frac{2}{k},
\end{align}
and the desired claim follows by considering all possible $k$.
\end{proof}
Combining Eq.~\eqref{Eq:Def-D_max} and Lemma~\ref{Lemma:AlternativeRepresentation}, we note the following:
\begin{align}\label{Eq:AlternativeFormDmaxR}
D^R_{\rm max}(\mathcal{N}\|\mE)&\coloneqq \sup_{\wt{\Lambda}_{\rm A},\rho_{\rm SA}} D_{\rm max}\left[(\mathcal{N}\otimes\widetilde{\Lambda}_{\rm A})(\rho_{\rm SA})\,\|\,(\mE\otimes\widetilde{\Lambda}_{\rm A})(\rho_{\rm SA})\right]\nonumber\\
&\coloneqq\log_2\sup_{\wt{\Lambda}_{\rm A},\rho_{\rm SA}}\inf\left\{\lambda\ge0\,|\,0\le[(\lambda\mE - \mathcal{N})\otimes\widetilde{\Lambda}_{\rm A}](\rho_{\rm SA})\right\}\nonumber\\
&=\log_2\inf\left\{\lambda\ge0\,|\,(\lambda\mE - \mathcal{N})\otimes\widetilde{\Lambda}_{\rm A}\;{\rm is\;a\;positive\;map}\;\forall{\rm A},\wt{\Lambda}_{\rm A}\right\}.
\end{align}
A direct observation from Eq.~\eqref{Eq:AlternativeFormDmaxR} is 
\begin{afact}\label{Eq:DmaxObservation}
$\left(2^{D^R_{\rm max}(\mathcal{N}\|\mE)}\mE - \mathcal{N}\right)\otimes\wt{\Lambda}_{\rm A}$ is a positive map $\forall{\rm A},\wt{\Lambda}_{\rm A}$.
\end{afact}
\begin{proof}
Suppose the opposite was correct.
Then there exists an ancillary system ${\rm A}_*$, an absolutely resource annihilating channel $\wt{\Lambda}_{\rm A_*}$, two states $\rho_{\rm SA_*}$, $\ket{\phi}$ such that $\bra{\phi}\left(2^{D^R_{\rm max}(\mathcal{N}\|\mE)}\mE - \mathcal{N}\right)\otimes\wt{\Lambda}_{\rm A_*}(\rho_{\rm SA_*})\ket{\phi}<0$.
However, we have $\bra{\phi}\left(2^{\left[D^R_{\rm max}(\mathcal{N}\|\mE)+\frac{1}{k}\right]}\mE - \mathcal{N}\right)\otimes\wt{\Lambda}_{\rm A_*}(\rho_{\rm SA_*})\ket{\phi}\ge0\;\forall k\in\mathbb{N}$ due to Eq.~\eqref{Eq:AlternativeFormDmaxR}.
This leads to a contradiction when $k\to\infty$.
\end{proof}

Finally, we note the following alternative form of Eq.~\eqref{Eq:P_D}:
\begin{align}\label{Eq:AlternativeForm}
P_{D_{\rm max}}(\mathcal{N}) =\log_2\inf_{\Lambda\in\mathcal{O}_R^N}\inf\left\{\lambda\ge0\,|\,(\lambda\Lambda - \mathcal{N})\otimes\widetilde{\Lambda}_{\rm A}\;{\rm is\;a\;positive\;map}\;\forall{\rm A},\wt{\Lambda}_{\rm A}\right\}.
\end{align}

Now, we can present the proofs of Theorem~\ref{Result:Upper-Bound} and Corollary~\ref{Coro:CCUpper-Bound}.

\subsection{Proof of Theorem~\ref{Result:Upper-Bound}}\label{App:MainProof-Result:Upper-Bound}
\begin{proof}
Consider a channel $\mathcal{N}'$ satisfying $\norm{\mathcal{N} - \mathcal{N}'}_\diamond\le2\delta$.
For a given error $\kappa>0$, recall that $\mathcal{O}_R^N(\kappa;\mathcal{N}')\coloneqq\{\Lambda\in\mathcal{O}_R^N\,|\,|D_{\rm max}^R(\mathcal{N}'\|\Lambda) - P_{D_{\rm max}}(\mathcal{N}')|\le\kappa\}$, which is by definition non-empty. 
Then for every $\Lambda^{\mathcal{N}'}\in\mathcal{O}_R^N(\kappa;\mathcal{N}')$, Fact~\ref{Eq:DmaxObservation} implies the existence of a positive map $\mathcal{P}$ such that
\begin{align}
&\mathcal{P}\otimes\wt{\Lambda}_{\rm A}\;{\rm is\;a\;positive\;map}\;\forall{\rm A},\wt{\Lambda}_{\rm A};\label{Eq:P_D-Relation2}\\
&\mathcal{N}' + \mathcal{P} = 2^{D^R_{\rm max}\left(\mathcal{N}'\|\Lambda^{\mathcal{N}'}\right)}\Lambda^{\mathcal{N}'}.\label{Eq:P_D-Relation3}
\end{align}
Note that the positivity of $\mathcal{P}$ actually follows from Eq.~\eqref{Eq:P_D-Relation2} and the fact that one is allowed to consider the trivial ancillary system, i.e., the case when there is no ancillary system.
Now, with a given $M$-code $\Theta_M = (\{\rho_m\}_{m=0}^{M-1},\{E_m\}_{m=0}^{M-1})$, we have [recall the definition from Eq.~\eqref{Eq:SuccessProbability}]: 
\begin{align}
p_s\left(\Theta_M,\mathcal{N}'\right)&\coloneqq\frac{1}{M}\sum_{m=0}^{M-1}{\rm tr}\left[E_m\mathcal{N}'(\rho_m)\right]\nonumber\\
&=\frac{2^{D^R_{\rm max}\left(\mathcal{N}'\|\Lambda^{\mathcal{N}'}\right)}}{M}\sum_{m=0}^{M-1}{\rm tr}\left[E_m\Lambda^{\mathcal{N}'}(\rho_m)\right]-\frac{1}{M}\sum_{m=0}^{M-1}{\rm tr}\left[E_m\mathcal{P}(\rho_m)\right]\nonumber\\
&\le\frac{2^{\left[P_{D_{\rm max}}(\mathcal{N}') + \kappa\right]}}{M}\sup_{\Theta_{M'}}\sum_{m=0}^{M'-1}{\rm tr}\left[E'_m\Lambda^{\mathcal{N}'}(\rho'_m)\right],
\end{align}
where the facts that $\Lambda^{\mathcal{N}'}\in\mathcal{O}_R^N(\kappa,\mathcal{N}')$ and ${\rm tr}\left[E_m\mathcal{P}(\rho_m)\right]\ge0$ for all $m$ imply the third line, and the maximization $\sup_{\Theta_{M'}}$ is taken over every $M'\in\mathbb{N}$ and $M'$-code $\Theta_{M'} = (\{\rho'_m\}_{m=0}^{M'-1},\{E'_m\}_{m=0}^{M'-1})$. 
Since this is true for every $\Lambda^{\mathcal{N}'}\in\mathcal{O}_R^N(\kappa,\mathcal{N}')$, we conclude the following with Eq.~\eqref{Eq:GammaQuantity}:
\begin{align}
p_s\left(\Theta_M,\mathcal{N}'\right)\le\frac{1}{M}\times2^{\left[P_{D_{\rm max}}(\mathcal{N}')+\Gamma_\kappa(\mathcal{N}') +\kappa\right]}.
\end{align}
Now we use the estimate $\left|p_s\left(\Theta_M,\mathcal{N}'\right) - p_s\left(\Theta_M,\mathcal{N}\right)\right|\le\frac{1}{2}\norm{\mathcal{N}' - \mathcal{N}}_\diamond$~\cite{Takagi2019-3}, where $\norm{\mE}_\diamond\coloneqq\sup_{{\rm A},\rho_{\rm SA}}\norm{(\mE\otimes\mathcal{I}_{\rm A})(\rho_{\rm SA})}_1$ is the diamond norm.
This can be seen by the following computation
\begin{align}\label{Eq:SmoothCondition}
p_s\left(\Theta_M,\mathcal{N}'\right) - p_s\left(\Theta_M,\mathcal{N}\right) &= \frac{1}{M}\sum_{m=0}^{M-1}{\rm tr}\left[E_m(\mathcal{N}' - \mathcal{N})(\rho_m)\right]\nonumber\\
&\le\frac{1}{2M}\sum_{m=0}^{M-1}\norm{\mathcal{N}' - \mathcal{N}}_\diamond\nonumber\\
&=\frac{1}{2}\norm{\mathcal{N}' - \mathcal{N}}_\diamond,
\end{align}
which follows from the estimate $\sup_\rho\sup_{0\le E\le\id}2{\rm tr}[E(\mE' - \mE)(\rho)]\le\norm{\mE' - \mE}_\diamond$~\cite{Takagi2019-3} for arbitrary channels $\mE,\mE'$.
Gathering the above ingredients, we conclude that for every channel $\mathcal{N}'$ satisfying $\norm{\mathcal{N} - \mathcal{N}'}_\diamond\le2\delta$ and $M$-code $\Theta_M$ achieving $p_s\left(\Theta_M,\mathcal{N}\right)\ge1-\epsilon$, we have
\begin{align}
1-\epsilon&\le p_s\left(\Theta_M,\mathcal{N}\right)\nonumber\\
&\le p_s\left(\Theta_M,\mathcal{N}'\right) + \delta\nonumber\\
&\le\frac{1}{M}\times2^{\left[P_{D_{\rm max}}(\mathcal{N}')+\Gamma_\kappa(\mathcal{N}') +\kappa\right]} + \delta.
\end{align}
In other words, for every given $\epsilon,\delta\ge0\;\&\;0<\kappa<1$ satisfying $\epsilon+\delta<1$ we have
\begin{align}
C_{(1)}^\epsilon(\mathcal{N})&\le\log_2\frac{1}{1 - \epsilon - \delta} + \kappa + \inf_{\norm{\mathcal{N} - \mathcal{N}'}_\diamond\le2\delta}\left[P_{D_{\rm max}}(\mathcal{N}')+\Gamma_\kappa(\mathcal{N}')\right]\nonumber\\
&\le\log_2\frac{2^\kappa}{1 - \epsilon - \delta} +\inf_{\norm{\mathcal{N} - \mathcal{N}'}_\diamond\le2\delta}P_{D_{\rm max}}(\mathcal{N}') + \Gamma_\kappa^\delta(\mathcal{N}),
\end{align}
and the result follows.
\end{proof}
\subsection{Remark}\label{App:ExIsotropicStateAsymm}
Note that in some cases Eq.~\eqref{Eq:Result:CCUpper-Bound} can be simplified.
For instance, when the free states are isotropic states~\cite{Horodecki1999-2} (i.e., $R$ is asymmetry of the group $U\otimes U^*$), then $\Gamma_\kappa^{\delta}(\mathcal{N})\le \log_2\left(2\times\frac{d^2}{d^2-1}\right)$~\cite{footnote:IsotropicStates}. 
When $d\gg1$, Theorem~\ref{Result:Upper-Bound} implies $C_{\rm (1)}^\epsilon(\mathcal{N})\lesssim P_{D_{\rm max}}^\delta(\mathcal{N}) + \log_2\frac{1}{1-\epsilon-\delta} + 1$, and the additional degrees of freedom of $(U\otimes U^*)$-asymmetry allow performance better than isotropic states.
Another example is when $R$ is athermality, which implies Corollary~\ref{Coro:CCUpper-Bound} that will be proved in the following sub-section.

\subsection{Proof of Corollary~\ref{Coro:CCUpper-Bound}}
\begin{proof}
First, from Eq.~(80) in Ref.~\cite{Hsieh2020-1} we learn that $P_{D_{\rm max}|\gamma}(\mE) = \sup_\rho D_{\rm max}[\mE(\rho)\|\gamma]$.
This means that $\mE(\rho)\le2^{P_{D_{\rm max}|\gamma}(\mE)}\gamma\;\forall\rho$ and hence
\begin{align}
\Gamma_\kappa(\mathcal{N}) &\coloneqq \log_2\inf_{\Lambda^\mathcal{N}\in\mathcal{O}_R^N(\kappa;\mathcal{N})}\sup_{\Theta_M}\sum_{m=0}^{M-1}{\rm tr}\left[E_m\Lambda^\mathcal{N}(\rho_m)\right]\nonumber\\
&\le\log_2\inf_{\Lambda^\mathcal{N}\in\mathcal{O}_R^N(\kappa;\mathcal{N})}2^{P_{D_{\rm max}|\gamma}\left(\Lambda^\mathcal{N}\right)}\sup_{\Theta_M}\sum_{m=0}^{M-1}{\rm tr}\left(E_m\gamma\right)\nonumber\\
& = \inf_{\Lambda^\mathcal{N}\in\mathcal{O}_R^N(\kappa;\mathcal{N})}P_{D_{\rm max}|\gamma}\left(\Lambda^\mathcal{N}\right),
\end{align}
and the result follows.
\end{proof}

\section{Proof of Theorem~\ref{Result:Asymmetry-Lower-Bounds}}\label{App:Proof-Result:Asymmetry-Lower-Bounds}
\begin{proof}
We follow the proof of Theorem 2 in Ref.~\cite{Korzekwa2019}.
First, a {\em $G$-twirling channel}, which is an operation used to symmetrize all input states with respect to a unitary group $G\coloneqq\left\{U^{(g)}\right\}_{g=1}^{|G|}$, is defined by
\begin{align}
\mathcal{T}_G(\cdot)\coloneqq\frac{1}{|G|}\sum_{g=1}^{|G|} U^{(g)}(\cdot)U^{(g),\dagger}.
\end{align}
When the group is infinite, one can replace the summation by integration equipped with the Haar measure:
$
\mathcal{T}_G(\cdot)\coloneqq\int_{U\in G} U(\cdot)U^\dagger dU.
$
We focus on the finite case to illustrate the proof.

With a given state $\rho$ and a given {\em codebook} $\mathcal{C}$ (that is, a mapping, $m\mapsto g_m$, from the classical information $\{m\}_{m=0}^{M-1}$ to the set $\{1,2,...,|G|\}$)~\cite{Korzekwa2019}, consider the encoding 
\begin{align}\label{Eq:LowerBoundEncoding}
\left\{\sigma_{g_m|\rho}\coloneqq U^{(g_m)}\rho U^{(g_m),\dagger}\right\}_{m=0}^{M-1}.
\end{align} 
To construct the decoding, consider the following $M$ elements of POVM (which is a pretty good measurement scheme):
\begin{align}\label{Eq:LowerBoundDecoding}
\left\{E_m^{\mathcal{C}|\rho}\coloneqq S\sigma_{g_m|\rho}S\right\}_{m=0}^{M-1},
\end{align}
 where $S\coloneqq\left(\sum_{m=0}^{M-1}\sigma_{g_m|\rho}\right)^{-\frac{1}{2}}$.
Note that for a positive semi-definite operator $A$, the notation $A^{-1}$ is the inverse of $A$ restricted to the support of $A$~\cite{Beigi2014}.
This means $A^{-1}A = AA^{-1}$ will be the projection onto the support of $A$, and we have $A^{-1}A = AA^{-1}\le\id$ in general.
Hence, $\left\{E_m^{\mathcal{C}|\rho}\right\}_{m=0}^{M-1}$ is not a POVM in general, since $\sum_{m=0}^{M-1}E_m^{\mathcal{C}|\rho}$ will be the projection onto the support of $\sum_{m=0}^{M-1}\sigma_{g_m|\rho}$.
Recently, Korzekwa {\em et al.} (see Eqs. (44), (45) and (51) in Ref.~\cite{Korzekwa2019}) show that for $0<\kappa<1$ we have
\begin{align}\label{Eq:KorzekwaResult}
\mathbb{E}_{\mathcal{C}}\frac{1}{M}\sum_{m=0}^{M-1}{\rm tr}\left[E_m^{\mathcal{C}|\rho}\sigma_{g_m|\rho}\right]\ge (1-\kappa)\left(1 - Me^{-D_s^\kappa[\rho\,\|\,\mathcal{T}_G(\rho)]}\right),
\end{align}
where $\mathbb{E}_{\mathcal{C}}$ indicates the average over randomly chosen codebook $\mathcal{C}$ (following Ref.~\cite{Korzekwa2019}, each $m$ is independently and uniformly at random encoded into the integer $g_m$, which means $\{g_m\}_{m=0}^{M-1}$ can be interpreted as independent and identically distributed random variables with uniform distribution).
For a $G$-covariant channel $\mathcal{N}$, consider the $M$-code given by $\Theta_{M}^{\mathcal{C}|\rho}\coloneqq\left(\{\sigma_{g_m|\rho}\}_{m=0}^{M-1},\{\wt{E}_m\}_{m=0}^{M-1}\right)$, where $\wt{E}_m\coloneqq E_m^{\mathcal{C}|\mathcal{N}(\rho)}$ for $m>0$ and $\wt{E}_0\coloneqq E_0^{\mathcal{C}|\mathcal{N}(\rho)} + \left(\id - \sum_{m=0}^{M-1}E_m^{\mathcal{C}|\mathcal{N}(\rho)}\right)$.
Note that $\sum_{m=0}^{M-1}E_m^{\mathcal{C}|\mathcal{N}(\rho)}$ is the projection onto the support of $\sum_{m=0}^{M-1}\sigma_{g_m|\mathcal{N}(\rho)}$, which means $\sum_{m=0}^{M-1}E_m^{\mathcal{C}|\mathcal{N}(\rho)}\le\id$.
Then we have
\begin{align}
\sup_\rho\mathbb{E}_{\mathcal{C}}p_s\left(\Theta_{M}^{\mathcal{C}|\rho},\mathcal{N}\right) &=\sup_\rho\mathbb{E}_{\mathcal{C}}\frac{1}{M}\left({\rm tr}\left[\left(\wt{E}_0 - E_0^{\mathcal{C}|\mathcal{N}(\rho)}\right)\mathcal{N}\left(\sigma_{g_m|\rho}\right)\right] + \sum_{m=0}^{M-1}{\rm tr}\left[E_m^{\mathcal{C}|\mathcal{N}(\rho)}\mathcal{N}\left(\sigma_{g_m|\rho}\right)\right]\right)\nonumber\\
&\ge \sup_\rho\mathbb{E}_{\mathcal{C}}\frac{1}{M}\sum_{m=0}^{M-1}{\rm tr}\left[E_m^{\mathcal{C}|\mathcal{N}(\rho)}\mathcal{N}\left(\sigma_{g_m|\rho}\right)\right]\nonumber\\
&=\sup_\rho\mathbb{E}_{\mathcal{C}}\frac{1}{M}\sum_{m=0}^{M-1}{\rm tr}\left[E_m^{\mathcal{C}|\mathcal{N}(\rho)}\sigma_{g_m|\mathcal{N}(\rho)}\right]\nonumber\\
&\ge(1-\kappa)\left(1 - Me^{-\sup_\rho D_s^\kappa[\mathcal{N}(\rho)\,\|\,\mathcal{T}_G\circ\mathcal{N}(\rho)]}\right).
\end{align}
Since $\wt{E}_0 - E_0^{\mathcal{C}|\mathcal{N}(\rho)} = \id - \sum_{m=0}^{M-1}E_m^{\mathcal{C}|\mathcal{N}(\rho)}$ is non-negative, the second line follows.
The third line is because $\mathcal{N}$ is $G$-covariant and $U^{(g_m)}\in G$, so we have $\mathcal{N}(\sigma_{g_m|\rho}) = \mathcal{N}(U^{(g_m)}\rho U^{(g_m),\dagger}) = U^{(g_m)}\mathcal{N}\left(\rho\right)U^{(g_m),\dagger} = \sigma_{g_m|\mathcal{N}(\rho)}$.
The last line is a direct application of Eq.~\eqref{Eq:KorzekwaResult} by replacing the role of $\rho$ by $\mathcal{N}(\rho)$.
This means when $1-\epsilon < (1-\kappa)\left(1 - Me^{-\sup_\rho D_s^\kappa[\mathcal{N}(\rho)\,\|\,\mathcal{T}_G\circ\mathcal{N}(\rho)]}\right)$, there must exist an $M$-code $\Theta_{M}^{\mathcal{C}|\rho}$ with some $\rho$ and $\mathcal{C}$ achieving $p_s(\Theta_{M}^{\mathcal{C}|\rho},\mathcal{N})\ge 1-\epsilon$.
Let $\log_2M_* = C_{\rm (1)}^{\epsilon}(\mathcal{N})$.
Because no $(M_*+1)$-code can achieve success probability $p_s\ge1-\epsilon$, we must have
\begin{align}\label{Eq:ComputationStep}
1-\epsilon \ge (1-\kappa)\left(1 - (M_*+1)e^{-\sup_\rho D_s^\kappa[\mathcal{N}(\rho)\,\|\,\mathcal{T}_G\circ\mathcal{N}(\rho)]}\right).
\end{align}
Following Ref.~\cite{Korzekwa2019}, we set $\kappa = \epsilon-\delta$.
Since $\log_2n\ge\log_2(n+1) - 1$ for all positive integer $n$, we conclude that
\begin{align}
\log_2{M_*}&> \frac{1}{\ln2}{\sup_\rho D_s^{\epsilon - \delta}[\mathcal{N}(\rho)\,\|\,\mathcal{T}_G\circ\mathcal{N}(\rho)]} + \log_2\delta -1 \nonumber\\
&\ge\frac{1}{\ln2}{\inf_{\Lambda\in\mathcal{O}_R^N}\sup_\rho D_s^{\epsilon - \delta}[\mathcal{N}(\rho)\,\|\,\Lambda(\rho)]} + \log_2\delta -1,
\end{align}
where the first line is a direct consequence of Eq.~\eqref{Eq:ComputationStep}, and the second line is because $\mathcal{T}_G\circ\mathcal{N}$ is an $G$-covariant channel that can only output symmetric states.
\end{proof}

\section{Collision Model for Thermalization}\label{App:ThermalizationModel}
Following Ref.~\cite{Sparaciari2019}, consider a finite dimensional system {\rm S} with Hamiltonian $H_{\rm S}$ and a bath ${\rm B}$, which is assumed to be $n-1$ copies of the system ${\rm S}$, and each copy has Hamiltonian $H_{\rm S}$.
Labeling the bath as ${\rm B}_1{\rm B_2}...{\rm B}_{n-1}$, this means $H_{\rm B} = \sum_{k=1}^{n-1}\id_{{\rm B}_1}\otimes...\otimes\id_{{\rm B}_{k-1}}\otimes H_{\rm S}\otimes\id_{{\rm B}_{k+1}}\otimes...\otimes\id_{{\rm B}_{n-1}}$.
The collision model introduced in Ref.~\cite{Sparaciari2019} used to depict thermalization processes is given by
\begin{align}\label{Eq:Model}
\frac{\partial\rho_{\rm SB}(t)}{\partial t} = \sum_k\lambda_k\left[U_{\rm SB}^{(k)}\rho_{\rm SB}(t)U_{\rm SB}^{(k),\dagger} - \rho_{\rm SB}(t)\right],
\end{align}
where $\rho_{\rm SB}(t)$ is the global state on ${\rm SB}$ at time $t$, $U_{\rm SB}^{(k)}$ represents an energy-preserving unitary on ${\rm SB}$; i.e., $[U_{\rm SB}^{(k)},H_{\rm S}+H_{\rm B}] = 0$, and $\lambda_k$ is the rate for $U_{\rm SB}^{(k)}$ to occur (see also Eqs.~(A2) and~(A3) in Appendix A of Ref.~\cite{Sparaciari2019}).
Roughly speaking, each $U_{\rm SB}^{(k)}$ models an elastic collision between certain subsystems of ${\rm SB}$.
We refer the reader to Ref.~\cite{Sparaciari2019} for the details of the model and its physical reasoning.
In this work, we use the notation $\mathcal{C}_n$ to denote the set of channels $\mE$ on ${\rm SB}$ such that, for every $\rho$ in ${\rm S}$, $\mE\left(\rho\otimes\gamma^{\otimes(n-1)}\right)$ can be realized by Eq.~\eqref{Eq:Model} at a time point $t>0$ [i.e., $\rho_{\rm SB}(t) = \mE\left(\rho\otimes\gamma^{\otimes(n-1)}\right)$; note that we also allow $t=\infty$] with $\rho_{\rm SB}(0) = \rho\otimes\gamma^{\otimes(n-1)}$. 
See also Ref.~\cite{Hsieh2020-1} for its connection with resource preservability.

\section{Proof of Theorem~\ref{Coro:LinkCoro}}\label{App:Proof-Coro:LinkCoro}
Theorem~\ref{Coro:LinkCoro} is a consequence of the combination of Theorem~\ref{Result:Upper-Bound} and Theorem 4 in Ref.~\cite{Hsieh2020-1}, which is formally stated as follows (here we implicitly assume the system Hamiltonian is the one realizing the thermal state $\gamma$ with some temperature):
\begin{atheorem}\label{AThm}{\em\cite{Hsieh2020-1}}
Given a Gibbs-preserving channel $\mathcal{N}$, $0\le\epsilon<1$, and a full-rank thermal state $\gamma$.
If $\mathcal{N}$ is coherence-annihilating and the system Hamiltonian satisfies the energy subspace condition, then we have
\begin{align}
2^{P_{D_{\rm max}|\gamma}}(\mathcal{N})\le\mathcal{B}^\epsilon_\gamma(\mathcal{N}) + \frac{2\sqrt{\epsilon}}{p_{\rm min}(\gamma)} + 1,
\end{align}
where $p_{\rm min}(\gamma)$ is the smallest eigenvalue of $\gamma$.
\end{atheorem}
We remark that being coherence-annihilating is required by the proof given in Ref.~\cite{Sparaciari2019} (specifically, it is crucial for the proof of Lemma 17 in Appendix C of Ref.~\cite{Sparaciari2019}), which explains the assumption made in Theorem~\ref{Coro:LinkCoro}.
Combining Corollary~\ref{Coro:CCUpper-Bound} (and hence Theorem~\ref{Result:Upper-Bound}) and Theorem~\ref{AThm}, we are now in the position to prove Theorem~\ref{Coro:LinkCoro}:
\begin{proof}
Consider coherence as the state resource ($R={\rm Coh}$) in Corollary~\ref{Coro:CCUpper-Bound}, we obtain
\begin{align}
C_{\rm (1)}^\epsilon(\mathcal{N})\le P_{D_{\rm max}|{\rm Coh}}(\mathcal{N})  + P_{D_{\rm max}|\gamma}\left(\Lambda^\mathcal{N}\right) + \log_2\frac{2^\kappa}{1-\epsilon}
\end{align}
for every $\Lambda^\mathcal{N}\in\mathcal{O}_{\rm Coh}^N(\kappa,\mathcal{N})$.
Hence, $\Lambda^\mathcal{N}$ is coherence-annihilating.
Use Theorem~\ref{AThm} (consequently, we need to assume all conditions made in the statement of Theorem~\ref{AThm}), we conclude that for $0\le\epsilon,\delta<1$ we have
\begin{align}
C_{\rm (1)}^\epsilon(\mathcal{N})\le P_{D_{\rm max}|{\rm Coh}}(\mathcal{N}) + \log_2\left(\mathcal{B}^{\delta}_\gamma(\Lambda^\mathcal{N}) + \frac{2\sqrt{\delta}}{p_{\rm min}(\gamma)} + 1\right) +  \log_2\frac{2^\kappa}{1 - \epsilon},
\end{align}
which completes the proof.
\end{proof}

\section{Proof of Theorem~\ref{Result:ConverseBound}}\label{App:Proof-Result:ConverseBound}
Before the proof, let us recall a crucial tool called {\em fully entangled fraction} (FEF)~\cite{Horodecki1999-2,Albeverio2002}. For a bipartite state $\rho$ with equal local dimension $d$, its FEF is defined by
\begin{align}
\F{\rho}\coloneqq\max_{\ket{\Phi_d}}\bra{\Phi_d}\rho\ket{\Phi_d},
\end{align}
which maximizes over all maximally entangled states $\ket{\Phi_d}$ with local dimension $d$.
FEF is well-known for characterizing different forms of inseparability~\cite{Horodecki1999-2,Albeverio2002,Zhao2010,Cavalcanti2013,Bell-RMP,Quintino2016,Hsieh2016,Hsieh2018E,Ent-RMP,Liang2019,Hsieh2020-2}. 
For instance, $\mathcal{F}_{\rm max}(\rho)>\frac{1}{d}$ implies $\rho$ is free entangled~\cite{Ent-RMP,Horodecki1999-1}, useful for teleporation~\cite{Horodecki1999-2}, multi-copy nonlocal~\cite{Palazuelos2012, Cavalcanti2013}, and multi-copy steerable~\cite{Hsieh2016,Quintino2016} (see also Appendix~\ref{App:LOSR-Nonlocality-Steering}).

Following the proof of Theorem~\ref{Result:Upper-Bound}, we will prove a lemma and have the main theorem as a corollary.
Given a state resource $R$, recall from Sec.~\ref{Sec:Formulation} that $\mathbb{F}_R$ is the set of free operations of $R$-preservability given by~\cite{Hsieh2020-1} $\mE\mapsto\Lambda_+\circ(\mE\otimes\wt{\Lambda})\circ\Lambda_-$, where $\Lambda_+,\Lambda_-\in\mathcal{O}_R$ are free operations and $\wt{\Lambda}\in\wt{\mathcal{O}}_R^N$ is an absolutely resource annihilating channel.
Since now, we will always assume that for every $\mathfrak{F}\in\mathbb{F}_R$ and $\mathcal{N}\in\mathcal{O}_R$, the input and output systems of $\mathfrak{F}(\mathcal{N})$ are both bipartite with finite equal local dimension.
Also, we will use the notation $\mathcal{F}_R^{(d\times d)}$ to denote the set of free states of $R$ in bipartite systems with equal local dimension $d$.

\begin{alemma}\label{Lemma:MESUpperBound}
Given $\epsilon,\delta>0$ satisfying $\epsilon+\delta<1$.
For $\mathcal{N}\in\mathcal{O}_R$ and $\mathfrak{F}\in\mathbb{F}_R$, if $M$ achieves $p_{s|{\rm ME}}\left[M,\mathfrak{F}(\mathcal{N})\right]\ge1-\epsilon$, then we have
\begin{align}
0\le P_{D_{\rm max}}^\delta(\mathcal{N}) + \log_2\frac{1}{1 - \epsilon - \delta} + \log_2F^R(d),
\end{align}
where 
\begin{align}
F^R(d)\coloneqq\sup_{\eta\in\mathcal{F}_R^{(d\times d)}}\mathcal{F}_{\rm max}(\eta),
\end{align}
and $d$ is the local dimension of the output bipartite system of $\mathfrak{F}(\mathcal{N})$.
\end{alemma}
\begin{proof}
For every positive integer $k$ and every $\mathcal{N}'$ such that $\norm{\mathcal{N}' - \mathcal{N}}_\diamond\le2\delta$, Eq.~\eqref{Eq:AlternativeForm} implies the existence of a value $\lambda_k\ge0$ and a $R$-annihilating channel $\Lambda_k\in\mathcal{O}_R^N$ such that (i) $\left|P_{D_{\rm max}}(\mathcal{N}') - \log_2\lambda_k\right|\le\frac{1}{k}$, and (ii) $\mathcal{P}_k\otimes\wt{\Lambda}_{\rm A}$ is a positive map for every ancillary system ${\rm A}$ and $\wt{\Lambda}_{\rm A}\in\wt{\mathcal{O}}_R^N$, where $\mathcal{N}' + \mathcal{P}_k = \lambda_k\Lambda_k$.
Write $\mathfrak{F}(\mE) = \Lambda_+\circ(\mE\otimes\wt{\Lambda})\circ\Lambda_-$ with $\Lambda_+,\Lambda_-\in\mathcal{O}_R$ and $\wt{\Lambda}\in\wt{\mathcal{O}}_R^N$.
This means for every positive integer $M$ we have (note that $\ket{\Phi'_m}$'s are all staying in the output space of $\mathfrak{F}(\mathcal{N})$, which is a bipartite system with equal local dimension $d$):
\begin{align}\label{Eq:MESComputation}
p_{s|{\rm ME}}[M,\mathfrak{F}(\mathcal{N}')]&\coloneqq\sup_{\{\ket{\Phi_m}\}_{m=0}^{M-1},\{\ket{\Phi'_m}\}_{m=0}^{M-1}}\frac{1}{M}\sum_{m=0}^{M-1}\bra{\Phi'_m} \Lambda_+\circ(\mathcal{N}'\otimes\wt{\Lambda})\circ\Lambda_-(\proj{\Phi_m})\ket{\Phi'_m}\nonumber\\
&\le\frac{2^{\left[P_{D_{\rm max}}(\mathcal{N}') + \frac{1}{k}\right]}}{M}\times\sup_{\{\ket{\Phi_m}\}_{m=0}^{M-1},\{\ket{\Phi'_m}\}_{m=0}^{M-1}}\sum_{m=0}^{M-1}\bra{\Phi'_m}\Lambda_+\circ(\Lambda_k\otimes\wt{\Lambda})\circ\Lambda_-(\proj{\Phi_m})\ket{\Phi'_m}\nonumber\\
&\le\frac{2^{\left[P_{D_{\rm max}}(\mathcal{N}') + \frac{1}{k}\right]}}{M}\times\sup_{\{\eta_m\}_{m=0}^{M-1}\subseteq\mathcal{F}_R^{(d\times d)},\{\ket{\Phi'_m}\}_{m=0}^{M-1}}\sum_{m=0}^{M-1}\bra{\Phi'_m}\eta_m\ket{\Phi'_m}\nonumber\\
&\le 2^{\left[P_{D_{\rm max}}(\mathcal{N}') + \frac{1}{k}\right]}\times F^R(d).
\end{align}
The second line is because $\Lambda_+\circ(\mathcal{P}_k\otimes\wt{\Lambda})\circ\Lambda_-$ is a positive map.
The third line is because $\Lambda_+\circ(\Lambda_k\otimes\wt{\Lambda})\circ\Lambda_-\in\mathcal{O}_R^N$, and the fact that the output system is a bipartite system with equal local dimension $d$.
From here we conclude that 
\begin{align}
p_{s|{\rm ME}}[M,\mathfrak{F}(\mathcal{N}')]\le 2^{P_{D_{\rm max}}(\mathcal{N}')}\times F^R(d).
\end{align}

Hence, for every $M$ achieving $p_{s|{\rm ME}}[M,\mathfrak{F}(\mathcal{N}')]\ge1-\epsilon$ and for every $\mathcal{N}'$ achieving $\norm{\mathcal{N}' - \mathcal{N}}_\diamond\le2\delta$, we have
\begin{align}
1-\epsilon&\le p_{s|{\rm ME}}[M,\mathfrak{F}(\mathcal{N})]\nonumber\\
&\le p_{s|{\rm ME}}[M,\mathfrak{F}(\mathcal{N}')] + \delta\nonumber\\
&\le 2^{P_{D_{\rm max}}(\mathcal{N}')}\times F^R(d) + \delta.
\end{align}
The second line is a consequence of Eqs.~\eqref{Eq:SmoothCondition} and~\eqref{Eq:EstimateAncillaryLambda}.
The desired result follows.
\end{proof}

As a direct observation on Lemma~\ref{Lemma:MESUpperBound}, once $F^R(d)$ has an explicit dependency on $d$, one could conclude an upper bound on $\log_2M$.
This is the case for various resources, and this fact allows us to prove Theorem~\ref{Result:ConverseBound} as follows.
\begin{proof}
In the first case, consider $R=$ athermality with the thermal state $\gamma$, which is in a bipartite system with equal local dimension.
Then it suffices to notice that in this case Eq.~\eqref{Eq:MESComputation} becomes 
\begin{align}
p_{s|{\rm ME}}[M,\mathfrak{F}(\mathcal{N}')]&\le\frac{2^{\left[P_{D_{\rm max}}(\mathcal{N}') + \frac{1}{k}\right]}}{M}\times\sup_{\{\ket{\Phi'_m}\}_{m=0}^{M-1}}\sum_{m=0}^{M-1}\bra{\Phi'_m}\gamma\ket{\Phi'_m}\nonumber\\
&\le\frac{2^{\left[P_{D_{\rm max}}(\mathcal{N}') + \frac{1}{k}\right]}}{M}.
\end{align}
This means we have $C_{{\rm ME}, (1)}^\epsilon[\mathfrak{F}(\mathcal{N})]\le P_{D_{\rm max}}^\delta(\mathcal{N}) + \log_2\frac{1}{1 - \epsilon - \delta}$ for all $\mathfrak{F}\in\mathbb{F}_R$, and the desired bound follows.

Now we recall that a bipartite state $\rho$ with equal finite local dimension $d$ is free entangled~\cite{Ent-RMP,Horodecki1998}, multi-copy nonlocal~\cite{Palazuelos2012,Cavalcanti2013}, and multi-copy steerable~\cite{Hsieh2016,Quintino2016} if $\mathcal{F}_{\rm max}(\rho)>\frac{1}{d}$.
This means for these resources we have $F^R(d)\le\frac{1}{d}$.
Given $\mathfrak{F}\in\mathbb{F}_R$ and $\mathcal{N}\in\mathcal{O}_R$, Lemma~\ref{Lemma:MESUpperBound} implies that for every $M$ achieving $p_{s|{\rm ME}}\left[M,\mathfrak{F}(\mathcal{N})\right]\ge1-\epsilon$, we have [in what follows we again use $d$ to denote the local dimension of the output bipartite system of $\mathfrak{F}(\mathcal{N})$]
\begin{align}
0&\le P_{D_{\rm max}}^\delta(\mathcal{N}) + \log_2\frac{1}{1 - \epsilon - \delta} + \log_2F^R(d)\nonumber\\
&\le P_{D_{\rm max}}^\delta(\mathcal{N}) + \log_2\frac{1}{1 - \epsilon - \delta}-\log_2\sqrt{M},
\end{align}
where the second line is because for any such $M$ the output space of $\mathfrak{F}(\mathcal{N})$ contains $M$ mutually orthonormal maximally entangled states, which means $M\le d^2$ and hence $\frac{1}{d}\le\frac{1}{\sqrt{M}}$.
From here we conclude that
\begin{align}
\frac{1}{2}\times\log_2M\le P_{D_{\rm max}}^\delta(\mathcal{N}) + \log_2\frac{1}{1 - \epsilon - \delta},
\end{align}
which is the desired bound by considering all possible $M$ and $\mathfrak{F}\in\mathbb{F}_R$.
\end{proof}

\section{Implications of Theorem~\ref{Result:ConverseBound}}\label{App:Maintainability-MoreResults}
Theorem~\ref{Result:ConverseBound} gives further implications to superdense coding and also connects FEF and resource preservability.
We briefly summarize these remarks in this appendix.

\subsection{Maintaining Orthogonal Maximal Entanglement and Superdense Coding}
Theorem~\ref{Result:ConverseBound} allows an interpretation for superdense coding, and we briefly introduce the setup here to illustrate this.
Consider two agents, Alice and Bob, sharing a maximally entangled state $\ket{\Phi_0}\coloneqq\frac{1}{\sqrt{d}}\sum_{i=0}^{d-1}\ket{ii}$ with local dimension $d$.
First, Alice encodes the classical information $m$ in her local system (a {\em qudit}) by applying $U_m$, the unitary operator achieving $(U_m\otimes\id_{\rm B})\ket{\Phi_0} = \ket{\Phi_m}$ with a given set of orthogonal maximally entangled states $\{\ket{\Phi_m}\}_{m=0}^{M-1}$.
After this, she sends her qudit to Bob, and both Alice's and Bob's qudits undergo a dynamics modeled by a bipartite channel $\mathcal{N}$.
After receiving Alice's qudit, Bob decodes the classical information $m$ from $\mathcal{N}(\proj{\Phi_m})$ by a bipartite measurement.
We call such task a {\em $d$ dimensional superdense coding through $\mathcal{N}$}. 
It is the conventional superdesne coding when $\mathcal{N} = \mathcal{I}$, where Bob can apply a $d^2$ dimensional Bell measurement to perfectly decode $d^2$ classical data when only one qudit has been sent.
In general, different channels have different abilities to admit superdense coding, and $\sup_{\mathfrak{F}\in\mathbb{F}_R}C^\epsilon_{{\rm ME},(1)}\left[\mathfrak{F}(\mathcal{N})\right]$ is the highest amount of classical information allowed by a $d$ dimensional superdense coding through a channel $\mathcal{N}\in\mathcal{O}_R$ even with all possible assistance structures constrained by $R$, i.e., $\mathfrak{F}\in\mathbb{F}_R$.
In this sense, $\sup_{\mathfrak{F}\in\mathbb{F}_R}C^\epsilon_{{\rm ME},(1)}\left[\mathfrak{F}(\mathcal{N})\right]$ can be understood as the {\em superdense coding ability} of $\mathcal{N}$, and Theorem~\ref{Result:ConverseBound} estimates the optimal performance of superdense coding.

\subsection{Fully Entangled Fraction and Resource Preservability}
It is worth mentioning that the proof of Theorem~\ref{Result:ConverseBound} largely relies on FEF. 
Once an FEF threshold with an explicit dependency of local dimension exists, a result similar to Theorem~\ref{Result:ConverseBound} can be obtained.
For example, from Ref.~\cite{Hsieh2016} we learn that $\rho$ is (two-way) steerable if $\F{\rho}>\frac{d-1+\sqrt{d+1}}{d\sqrt{d+1}}$.
This means when $R=$ (two-way) steerability and $M\le d^2$, we have 
$
F^R(d)\le\frac{d-1+\sqrt{d+1}}{d\sqrt{d+1}}\le\frac{1}{\sqrt{d}} + \frac{1}{d}\le \frac{1}{M^{\frac{1}{4}}}+\frac{1}{\sqrt{M}}.
$
Hence, when $R=$ (two-way) steerability, we have 
\begin{align}
\frac{1}{4}\times\sup_{\mathfrak{F}\in\mathbb{F}_R}C^\epsilon_{{\rm ME},(1)}\left[\mathfrak{F}(\mathcal{N})\right]\le P_{D_{\rm max}}^\delta(\mathcal{N}) + \log_2\frac{1}{1 - \epsilon - \delta} + 1.
\end{align}
It turns out that resource preservability can be related to FEF as follows

\begin{aproposition}\label{Result:FEF-P_D}
Given a resource $R$, then $\mathcal{N}\in\mathcal{O}_R$ cannot maintain any maximally entangled state with an average error less than $\epsilon$ if
\begin{eqnarray}
P_{D_{\rm max}}(\mathcal{N}) < \log_2\frac{1-\epsilon}{F^R(d)},
\end{eqnarray}
where $d$ is the local dimension of the output bipartite space of $\mathcal{N}$.
\end{aproposition}
\begin{proof}
Applying Lemma~\ref{Lemma:MESUpperBound} with $\delta=0$, we learn that there exists no $M$ that can achieve $p_s(M,\mathcal{N})\ge1-\epsilon$ if 
\begin{align}
P_{D_{\rm max}}(\mathcal{N}) <-\log_2F^R(d) + \log_2(1 - \epsilon).
\end{align}
In other words, $\mathcal{N}\in\mathcal{O}_R$ cannot maintain any maximally entangled state with an average error lower than $\epsilon$ when this inequality is satisfied.
\end{proof}
Proposition~\ref{Result:FEF-P_D} gives a new way to understand FEF: Once a channel's resource preservability is not strong enough compared with a threshold induced by FEF, it is impossible to maintain maximal entanglement to the desired level.
For examples, suppose $\mathcal{N}$ is a free operation of free entanglement (or, similarly, multi-copy nonlocality or multi-copy steerability; note that we need to assume Assumptions~\ref{Postulate:Identity} and~\ref{Postulate:Composition}), then Proposition~\ref{Result:FEF-P_D} implies that $\mathcal{N}$ cannot maintain any maximally entangled state with an average error lower than $\epsilon$ if
$
P_{D_{\rm max}}(\mathcal{N}) < \log_2d(1-\epsilon).
$

\section{Proof of Theorem~\ref{Result:GeneralizedEPLT}}\label{App:Proof-Result:GeneralizedEPLT}
To demonstrate the proof, we will first provide a generalized version of the entanglement preserving local thermalization introduced in Ref.~\cite{Hsieh2020-2}.
After that, Theorem~\ref{Result:GeneralizedEPLT} can be proved by using it and the capacity-like measure $C_{\rm ME,(1)}$ defined in Eq.~\eqref{Eq:Capacity-Like-Measure}.
Before proceeding, we remark that, in this appendix, we will use subscripts to denote the located subsystems for both states and channels; e.g., $\rho_{\rm X}$ and $\mE_{\rm X}$ are a state and a channel in the system ${\rm X}$, respectively.
Also, following Ref.~\cite{Hsieh2020-2}, we assume no energy degeneracy for all subsystems, and all thermal states are assumed to be full-rank.

To start with, we recall from Ref.~\cite{Hsieh2020-2} the following alternative definition of a local thermalization with a given pair of thermal states $(\gamma_{\rm A},\gamma_{\rm B})$ on the subsystem ${\rm A,B}$, respectively [recall from Eq.~\eqref{Eq:LOSR} the definition of local operation and pre-shared randomness (LOSR) channels]:
\begin{adefinition}{\rm~\cite{Hsieh2020-2}}
A bipartite channel $\mE_{\rm AB}$ on ${\rm AB}$ is called a {\em local thermalization} to $(\gamma_{\rm A},\gamma_{\rm B})$ if
\begin{enumerate}
\item $\mE_{\rm AB}$ is a LOSR channel in ${\rm AB}$ bipartition.
\item ${\rm tr}_{\rm B}\circ\mE_{\rm AB}(\rho_{\rm AB}) = \gamma_{\rm A}$ and ${\rm tr}_{\rm A}\circ\mE_{\rm AB}(\rho_{\rm AB}) = \gamma_{\rm B}$ for all states $\rho_{\rm AB}$.
\end{enumerate}
\end{adefinition}
As remarked in Ref.~\cite{Hsieh2020-2}, a local thermalization is {\em local} in the sense that it is a {\em LOSR channel} that can {\em locally thermalize} all inputs.
This definition is equivalent to the one given in the main text, and we refer the reader to Definition 1, Definition 2, and Theorem 2 in Ref.~\cite{Hsieh2020-2} for the detailed reasoning.

Now, consider a tripartite system ${\rm ABC}$ with finite local dimensions $d,d,d^2+1$, respectively.
As mentioned in the main text, we will focus on the bipartition ${\rm A|BC}$, and the subsystem ${\rm C}$ can be treated as an ancillary system possessed by the local agent in ${\rm B}$.
In what follows, $\gamma_{{\rm X}}$ is the given thermal state of the subsystem ${\rm X=A,B,C}$ with temperature $T_{\rm X}$ and Hamiltonian $H_{\rm X}$.
Suppose $\{\ket{n}_{{\rm X}}\}_{n=0}^{D_{\rm X}-1}$ is the orthonormal energy eigenbasis of $H_{\rm X}$ with local dimension $D_{\rm X}$.
Then we define the following channel, which has a schematic interpretation given in Fig.~\ref{Fig:forEq:CCEPLT}.
First, consider a given ordered set of local unitary operators in ${\rm B}$ denoted by ${\bf V}\coloneqq\left\{V^{(n)}_{\rm B}\right\}_{n=0}^{d^2},$ where $V^{(d^2)}_{\rm B}\coloneqq\id_{\rm B}$ is always assumed to be the maximally mixed state, and other elements are arbitrary.
Then, for every $\kappa\in[0,1]$, we define
\begin{align}\label{Eq:CCEPLT}
\mE_{\rm A|BC}^{{\bf V},\kappa}\coloneqq\left(\mathcal{D}_{\rm AB}^\kappa\otimes\mathcal{I}_{\rm C}\right)\circ\left(\mathcal{I}_{\rm A}\otimes\mathcal{L}_{\rm BC}^{\bf V}\right)\circ\left(\mathcal{T}_{\rm AB}\otimes\mathcal{I}_{\rm C}\right),
\end{align}
whose components are defined as follows.
First, in a bipartite system with finite equal local dimension, 
\begin{align}
\mathcal{T}(\cdot)\coloneqq\int (U\otimes U^*)(\cdot)(U\otimes U^*)^\dagger dU
\end{align}
is the {\em $(U\otimes U^*)$-twirling operation}~\cite{Horodecki1999-1,Bennett1996} used to symmetrize the local states.
To encode classical information into the correlation shared by ${\rm A,B}$, we define the following local channel in ${\rm BC}$:
\begin{align}
\mathcal{L}_{\rm BC}^{\bf V}(\cdot)\coloneqq {\rm tr}_{\rm C}\left[\sum_{n=0}^{d^2}\left(V^{(n)}_{\rm B}\otimes\proj{n}_{\rm C}\right)(\cdot) \left(V^{(n)}_{\rm B}\otimes\proj{n}_{\rm C}\right)^\dagger\right]\otimes\gamma_{\rm C}.
\end{align}
Using operator-sum representation~\cite{QCI-book}, one can check that this is indeed a channel.
Finally, 
\begin{align}
\mathcal{D}_{\rm AB}^\kappa(\cdot)\coloneqq(1-\kappa)\neweta_{\rm A}^\kappa\otimes\neweta_{\rm B}^\kappa + \kappa\times(\cdot)
\end{align}
where
\begin{align}
\neweta_{\rm X}^\kappa\coloneqq\gamma_{\rm X} + \frac{\kappa}{1-\kappa}\left(\gamma_{\rm X} - \frac{\id_{\rm X}}{d}\right)
\end{align}
is an alternative thermal state for ${\rm X=A,B}$.
Let $p_{\rm min|AB}\coloneqq\min\{p_{\rm min}(\gamma_{{\rm A}});p_{\rm min}(\gamma_{{\rm B}})\}$, which is the smallest eigenvalue among $\gamma_{\rm A},\gamma_{\rm B}$.
Write $\gamma_{\rm BC}\coloneqq\gamma_{\rm B}\otimes\gamma_{\rm C}$,
then we have the following result:
\begin{atheorem}\label{Prop:LT}
For every $d<\infty$ and ${\bf V}$, we have
\begin{enumerate}
\item$\mE_{\rm A|BC}^{{\bf V},\kappa}$ is a local thermalization to $(\gamma_{\rm A},\gamma_{\rm BC})$ for all $0\le\kappa\le \kappa_*\coloneqq dp_{\rm min|AB}$.
\item$\mE_{\rm A|BC}^{{\bf V},\kappa = \kappa_*}$ is an entanglement preserving local thermalization to $(\gamma_{\rm A},\gamma_{\rm BC})$ for all full-rank $\gamma_{{\rm A}}$ and $\gamma_{{\rm B}}$.
\end{enumerate}
\end{atheorem}
\begin{proof}
Following Lemma 1 in Ref.~\cite{Hsieh2020-2}, we conclude that the local output of $\mE_{\rm A|BC}^{{\bf V},\kappa}$ in ${\rm A}$ and ${\rm BC}$ will always be $\gamma_{{\rm A}}$ and $\gamma_{{\rm BC}}$, respectively. 
Also, $0\le\kappa\le \kappa_*$ guarantees that both $\neweta_{{\rm A}}^\kappa,\neweta_{{\rm B}}^\kappa$ are legal quantum states. 
This implies that $\mE_{\rm A|BC}^{{\bf V},\kappa}$ is a legal channel, which is by definition a LOSR channel in the ${\rm A|BC}$ bipartition.
Hence, for all ${\bf V}$, $\mE_{\rm A|BC}^{{\bf V},\kappa}$ is a local thermalization when $\kappa$ is in the desired range.
Finally, when $\gamma_{{\rm A}}$ and $\gamma_{{\rm B}}$ are both full-rank (and hence non-pure), Theorem 1 in Ref.~\cite{Hsieh2020-2} implies that $\mE_{\rm A|BC}^{{\bf V},\kappa = \kappa_*}(\proj{\Psi_d^+}_{\rm AB}\otimes\proj{d^2}_{\rm C})$ will be entangled in the ${\rm A|BC}$ bipartition since in the subsystem ${\rm AB}$ it is identical to the entanglement preserving local thermalization constructed in Ref.~\cite{Hsieh2020-2}.
This proves the desired claim.
\end{proof}

\begin{center}
\begin{figure}
\scalebox{0.8}{\includegraphics{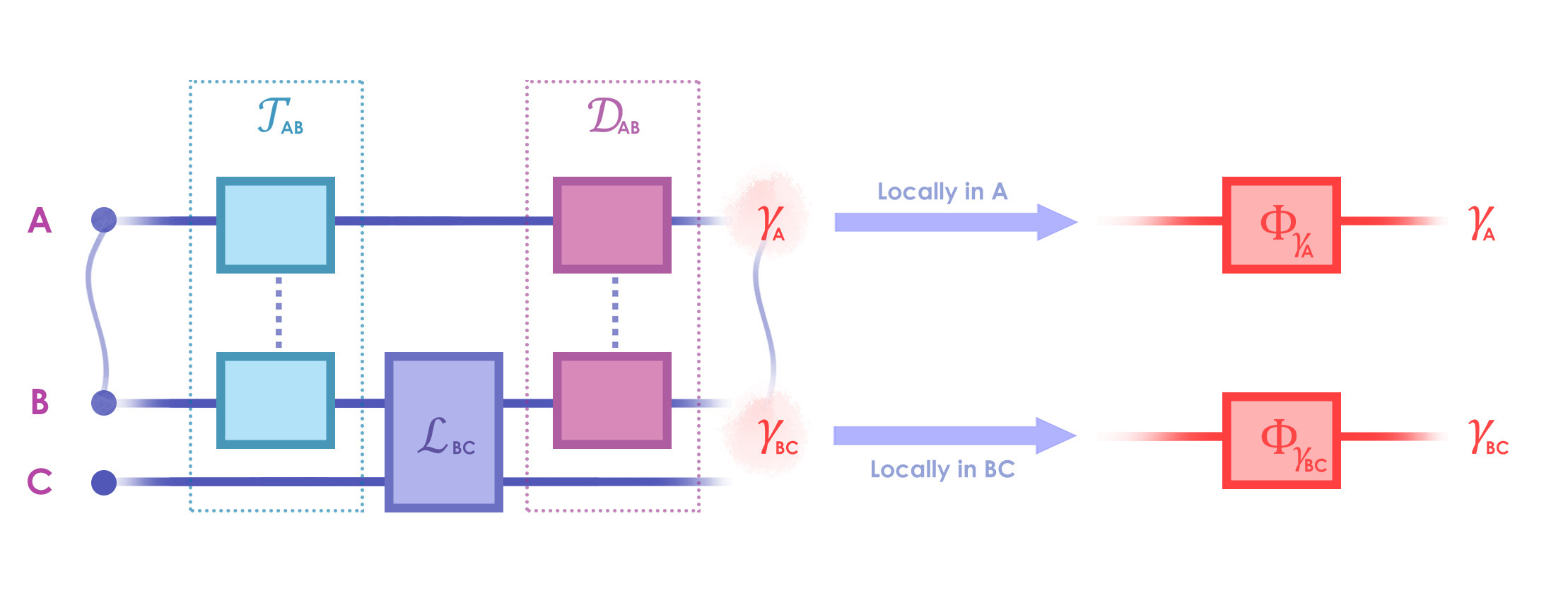} }
\caption{
Schematic interpretation of Eq.~\eqref{Eq:CCEPLT} and Theorem~\ref{Prop:LT}.
We use $\kappa = \kappa_*$ as an example.
As shown in the figure, the channel $\mE_{\rm A|BC}^{{\bf V},\kappa = \kappa_*}$ can be decomposed into three steps. 
First, a $(U\otimes U^*)$-twirling operation $\mathcal{T}_{\rm AB}$ is applied, which only requires local unitary operations and shared randomness (the thick dashed line).
After that, the channel $\mathcal{L}_{\rm BC} = \mathcal{L}_{\rm BC}^{\bf V}$ is applied locally in the subsystem ${\rm BC}$.
Finally, using shared randomness again, the channel $\mathcal{D}_{\rm AB} = \mathcal{D}_{\rm AB}^{\kappa_*}$ can be achieved.
The resulting global channel has the property that, when one only observes it locally in ${\rm A}$ and ${\rm BC}$, its is identical to a full thermalization channel; namely, a state preparation channel $\Phi_{\gamma_{\rm X}}(\cdot)\coloneqq\gamma_{\rm X}{\rm tr}(\cdot)$ of the given thermal state in the subsystem ${\rm X}$, where ${\rm X=A}$ or ${\rm BC}$.
Still, as shown in this section, this channel is able to preserve entanglement in the ${\rm AB}$ bipartition, and it is also possible to transmit classical information by encoding the information into the global correlation shared by ${\rm A}$ and ${\rm B}$.
}
\label{Fig:forEq:CCEPLT} 
\end{figure}
\end{center}

We leave a schematic interpretation in Fig.~\ref{Fig:forEq:CCEPLT}, and now we proceed to prove Theorem~\ref{Result:GeneralizedEPLT} by using Theorem~\ref{Prop:LT}:
\begin{proof}
Note that by assuming each subsystem Hamiltonian $H_{{\rm X}}$ to be non-degenerate and finite-energy, the corresponding thermal state $\gamma_{{\rm X}}$ is non-pure and full-rank if and only if $T_{{\rm X}}>0$~\cite{Hsieh2020-2}.
Hence, for every $T_{{\rm X}}>0$, we conclude that $\mE_{\rm A|BC}^{{\bf V},\kappa_*}$ is an entanglement preserving local thermalization to $(\gamma_{\rm A},\gamma_{\rm BC})$ according to Theorem~\ref{Prop:LT}.
Now we note that the one-shot classical capacity of $\mE_{\rm A|BC}^{{\bf V},\kappa_*}$ can be estimated by the ability of $\mathcal{D}_{\rm AB}^{\kappa_*}$ to maintain maximally entangled bases (see Sec.~\ref{Sec:MaintainingMaxEnt}).
More precisely, for a given $M\le d^2$ and sets of orthonormal maximally entangled states (in the ${\rm A|B}$ bipartition) $\left\{\ket{\Phi_m}\right\}_{m=0}^{M-1},\left\{\ket{\Phi'_m}\right\}_{m=0}^{M-1}$, the corresponding success probability of $\mathcal{D}_{\rm AB}^{\kappa_*}$ reads $\frac{1}{M}\sum_{m=0}^{M-1}\bra{\Phi'_m}\mathcal{D}_{\rm AB}^{\kappa_*}(\proj{\Phi_m})\ket{\Phi'_m}$.
Being maximally entangled, there exist $M$ unitary operators in ${\rm B}$, $\{U^{(m)}_{\rm B}\}_{m=0}^{M-1}$, such that $\ket{\Phi_m} = \left(\id_{\rm A}\otimes U^{(m)}_{\rm B}\right)\ket{\Psi_d^+}$, where $\ket{\Psi_d^+}\coloneqq\sum_{i=0}^{d-1}\frac{1}{\sqrt{d}}\ket{ii}$.
Define the set ${\bf U} = \{{U}_{\rm B}^{(m)}\}_{m=0}^{d^2}$ with ${U}^{(m)}_{\rm B} = \id_{\rm B}$ for all $m\ge M$.
Now we choose the encoding $\{\rho_m\}_{m=0}^{M-1}$ as
\begin{align}
\rho_{m}\coloneqq\proj{\Psi_d^+}_{\rm AB}\otimes\proj{m}_{\rm C}
\end{align} 
and decoding POVM $\left\{E_m\right\}_{m=0}^{M-1}$ as
\begin{align}
&E_m\coloneqq\proj{\Phi'_m}_{\rm AB}\otimes\id_{\rm C}\quad\forall\;m<M-1;\\
&E_{M-1}\coloneqq\id_{\rm ABC} - \sum_{m=0}^{M-2}E_m = \left(\id_{\rm AB} - \sum_{m=0}^{M-2}\proj{\Phi'_m}_{\rm AB}\right)\otimes\id_{\rm C}.
\end{align}
Then we have
\begin{align}
\frac{1}{M}\sum_{m=0}^{M-1}\bra{\Phi'_m}\mathcal{D}_{\rm AB}^{\kappa_*}(\proj{\Phi_m})\ket{\Phi'_m} &=\frac{1}{M}\sum_{m=0}^{M-1}{\rm tr}\left[\left(\proj{\Phi'_m}_{\rm AB}\otimes\id_{\rm C}\right)\mE_{\rm A|BC}^{{\bf U},\kappa_*}(\rho_m)\right]\nonumber\\
&\le\frac{1}{M}\sum_{m=0}^{M-1}{\rm tr}\left[E_m\mE_{\rm A|BC}^{{\bf U},\kappa_*}(\rho_m)\right],
\end{align}
where the inequality is due to the fact that $\proj{\Phi'_{M-1}}_{\rm AB}\otimes\id_{\rm C}\le E_{M-1}$.
This means [see Eq.~\eqref{Eq:Capacity-Like-Measure} for the definition of $C_{\rm ME,(1)}^\epsilon$]
\begin{align}\label{Eq:Computation-LowerBoundEPLT}
\sup_{{\bf V}}C_{\rm (1)}^\epsilon\left(\mE_{\rm A|BC}^{{\bf V},\kappa_*}\right)\ge C_{\rm ME,(1)}^\epsilon\left(\mathcal{D}_{\rm AB}^{\kappa_*}\right),
\end{align}
where the optimization is taken over every possible ordered set of $d^2+1$ unitary operators ${\bf V}\coloneqq\left\{V^{(n)}_{\rm B}\right\}_{n=0}^{d^2}$ with $V^{(d^2)}_{\rm B}\coloneqq\id_{\rm B}$.
On the other hand, we have
\begin{align}
\frac{1}{M}\sum_{m=0}^{M-1}\bra{\Phi_m}\mathcal{D}_{\rm AB}^{\kappa_*}(\proj{\Phi_m})\ket{\Phi_m} = \frac{1-\kappa_*}{M}\sum_{m=0}^{M-1}\bra{\Phi_m}\left(\neweta_{\rm A}^{\kappa_*}\otimes\neweta_{\rm B}^{\kappa_*}\right)\ket{\Phi_m} + \kappa_*.
\end{align}
Note that $\mathcal{D}_{\rm AB}^{\kappa_*}$ has output system dimension $d^2$.
Hence, when $M=d^2$, the success probability is given by $\frac{1-\kappa_*}{d^2} + \kappa_*.$
This means that $C_{\rm ME, (1)}^\epsilon\left(\mathcal{D}_{\rm AB}^{\kappa_*}\right) = \log_2d^2$ if $\frac{1-\kappa_*}{d^2} + \kappa_*\ge1-\epsilon$, or, equivalently,
\begin{align}\label{Eq:ComputationCondition-epsilon}
\epsilon\ge\left(1 - \frac{1}{d^2}\right)(1-\kappa_*).
\end{align}
Consequently, by using Eq.~\eqref{Eq:Computation-LowerBoundEPLT}, we learn that when Eq.~\eqref{Eq:ComputationCondition-epsilon} holds, there must exists an ${\bf V}$ such that $C_{\rm (1)}^\epsilon\left(\mE_{\rm A|BC}^{{\bf V},\kappa_*}\right)\ge\log_2d^2$.
The result follows.
\end{proof}


\begin{thebibliography}{99}
\bibitem{footnote:Local-Dimension} We implicitly assume that both input and output systems of $\mathcal{N}$ are bipartite systems with equal local dimensions. The local dimension for the input and output spaces can be different. Also, note that $\ket{\Phi'_m}$'s are all stay in the output system of the information processing channel.
\end{thebibliography}
\end{document}